\documentclass[preprint,12pt,authoryear]{elsarticle}

\usepackage{amssymb}
\usepackage{amsthm}

\usepackage{lineno}

\usepackage{amsmath,enumerate,mathrsfs,amsfonts,scalerel,subfigure}
\newtheorem{thm}{Theorem}
\newtheorem{cor}{Corollary}
\newtheorem{prop}{Proposition}
\newtheorem{lem}{Lemma}
\newtheorem{rem}{Remark}

\newcommand{\sref}[1]{Section~\ref{#1}}
\newcommand{\tref}[1]{Table~\ref{#1}}
\newcommand{\fref}[1]{Figure~\ref{#1}}
\newcommand{\cref}[1]{Chapter~\ref{#1}}

\newcommand{\thmref}[1]{Theorem~\ref{#1}}

\def\sD{\mathscr{D}}

\def\sK{\mathscr{K}}
\def\sL{\mathscr{L}}

\def\bbE{\mathbb{E}}
\def\bbF{\mathbb{F}}

\def\bbN{\mathbb{N}}

\def\bbP{\mathbb{P}}

\def\bbR{\mathbb{R}}

\def\gd{\delta}
\def\ge{\epsilon}

\def\gl{\lambda}

\def\gq{\theta}
\def\gr{\rho}

\def\gD{\Delta}

\def\gL{\Lambda}

\def\bfmath#1{\boldsymbol{#1}}

\def\bfe{{\bfmath{e}}}

\def\bfi{{\bfmath{i}}}
\def\bfj{{\bfmath{j}}}
\def\bfk{{\bfmath{k}}}
\def\bfl{{\bfmath{l}}}

\def\bfn{{\bfmath{n}}}

\def\bfs{{\bfmath{s}}}
\def\bft{{\bfmath{t}}}

\def\bfv{{\bfmath{v}}}

\def\bfx{{\bfmath{x}}}
\def\bfy{{\bfmath{y}}}
\def\bfz{{\bfmath{z}}}

\def\bfI{\bfmath{I}}

\def\bfL{\bfmath{L}}

\def\bfN{\bfmath{N}}

\def\bfP{\bfmath{P}}
\def\bfQ{\bfmath{Q}}

\def\bfX{\bfmath{X}}

\def\bfzero{\bfmath{0}}
\def\bfep{\bfmath{\epsilon}}
\def\bfxi{\bfmath{\xi}}
\def\bfpi{\bfmath{\pi}}

\def\bfTheta{\bfmath{\Theta}}
\def\ds{\ensuremath{\displaystyle}}
\def\ud{\mathrm{d}}
\def\Q{S} 

\def\Al{\ensuremath{A_{\leq l}}}
\def\Ag{\ensuremath{A_{> l}}}
\def\Alx{\ensuremath{A_{\leq x}}}
\def\Agx{\ensuremath{A_{> x}}}
\def\Ll{\ensuremath{[L]_{\leq l}}}
\def\Lg{\ensuremath{[L]_{> l}}}
\newcommand{\m}[3]{#1_{-#2,#3}}
\DeclareMathOperator*{\bigtimes}{\scalerel*{\times}{\sum}}

\makeatletter
\newcommand{\specialcell}[1]{\ifmeasuring@#1\else\omit$\displaystyle#1$\ignorespaces\fi} 
\makeatother

\journal{Theoretical Population Biology}

\begin{document}

\begin{frontmatter}



\title{A coalescent dual process for a Wright-Fisher diffusion with recombination and its application to haplotype partitioning}


\author[OxStat]{Robert C.\ Griffiths}
\author[WarStat,WarCS]{Paul A.\ Jenkins\corref{cor1}}
\author[Mon]{Sabin Lessard}
\address[OxStat]{Department of Statistics, University of Oxford, United Kingdom}
\address[WarStat]{Department of Statistics, University of Warwick, United Kingdom}
\address[WarCS]{Department of Computer Science, University of Warwick, United Kingdom}
\address[Mon]{D\'epartement de Math\'ematiques et de Statistique, Universit\'e de Montr\'eal, Montr\'eal, Canada}
\cortext[cor1]{Corresponding author. Address: Department of Statistics, University of Warwick, Coventry CV4 7AL, United Kingdom. E-mail: p.jenkins@warwick.ac.uk}
\begin{abstract}
Duality plays an important role in population genetics. It can relate results from forwards-in-time models of allele frequency evolution with those of backwards-in-time genealogical models; a well known example is the duality between the Wright-Fisher diffusion for genetic drift and its genealogical counterpart, the coalescent. There have been a number of articles extending this relationship to include other evolutionary processes such as mutation and selection, but little has been explored for models also incorporating crossover recombination. Here, we derive from first principles a new genealogical process which is dual to a Wright-Fisher diffusion model of drift, mutation, and recombination. The process is reminiscent of the \emph{ancestral recombination graph}, a widely-used multilocus genealogical model, but here ancestral lineages are typed and transition rates are regarded as being conditioned on an observed configuration at the leaves of the genealogy. Our approach is based on expressing a putative duality relationship between two models via their infinitesimal generators, and then seeking an appropriate test function to ensure the validity of the duality equation. This approach is quite general, and we use it to find dualities for several important variants, including both a discrete $L$-locus model of a gene and a continuous model in which mutation and recombination events are scattered along the gene according to continuous distributions. As an application of our results, we derive a series expansion for the transition function of the diffusion. Finally, we study in further detail the case in which mutation is absent. Then the dual process describes the dispersal of ancestral genetic material across the ancestors of a sample. The stationary distribution of this process is of particular interest; we show how duality relates this distribution to haplotype fixation probabilities. We develop an efficient method for computing such probabilities in multilocus models.
\end{abstract}

\begin{keyword}
coalescent \sep Wright-Fisher diffusion \sep recombination \sep duality
\end{keyword}
\end{frontmatter}

\section{Introduction}
\label{sec:introduction}
The concept of duality is a powerful technique for inferring the properties of one Markov process by looking at another related process, usually (as in this paper) discovered by considering the dynamics of the former in reverse time \citep[see][for recent review]{jan:kur:2014}. The idea has found many applications in population genetics, playing for example a central role in the constructions of the ancestral selection graph \citep{kro:neu:1997, neu:kro:1997} and the ancestral influence graph \citep{don:kur:1999:AAP}. One particularly well known duality is between the Wright-Fisher diffusion describing pure genetic drift and Kingman's coalescent \citep{kin:1982:SPA}. To illustrate the idea, consider a single neutral locus with two alleles. The Wright-Fisher diffusion $(X_t)_{t\geq 0}$ is the process on $[0,1]$ describing the evolution of the frequency of one allele, with infinitesimal generator
\begin{equation}
\label{eq:WFGenerator}
\sL f(x) = \frac{1}{2}x(1-x)f''(x)
\end{equation}
and domain $\sD(\sL) = C^2([0,1])$. The corresponding dual is the pure death process $(L_t)_{t\geq 0}$ on $\bbN = \{0,1,\dots\}$ with infinitesimal generator
\begin{equation}
\label{eq:KingmanGenerator}
\sK f(n) = \binom{n}{2}[f(n-1) - f(n)],
\end{equation}
which describes the dynamics of the \emph{ancestral}, or \emph{block-counting}, process of Kingman's coalescent. 

The two processes are dual with respect to the function $F: [0,1]\times \bbN \to \bbR$ defined by $F(x,n) = x^n$ (i.e.\ \emph{moment duals}): for each $x \in [0,1]$, $n \in \bbN$ and $t \geq 0$,
\begin{equation}
\label{eq:semigroupduality}
\bbE[F(X_t,n)\mid X_0 = x] = \bbE[F(x,L_t)\mid L_0 = n].
\end{equation}
We note for later use that this implies
\begin{equation}
\label{eq:duality}
\sL F(\cdot,n)(x) = \sK F(x,\cdot)(n), \qquad x\in[0,1], \quad n \in \bbN,
\end{equation}
and for general $\sL$, $\sK$, and $F$, the converse is also true under certain conditions on $F$ \citep{jan:kur:2014}. We also emphasise that, in this example and all others encountered in this paper, this duality is obtained via time-reversal, so that the time indices in the two processes run in different directions. Were we to run the two processes on a joint probability space, running $X_t$ from time 0 to $T$ would correspond to running $L_t$ \emph{backwards} from time $T$ to 0.

There have been numerous extensions to the models captured by \eqref{eq:duality}. For example, \citet{eth:gri:1990:AAP} describe a birth-death process which is dual to a two-locus Wright-Fisher diffusion with recombination between the two loci, and use it to prove an ergodic theorem for the diffusion. \citet{man:2013} uses the same process to derive a method to compute the transient moments of the diffusion. Generalising further, \citet{eth:kur:1993} describe a duality relationship between a Fleming-Viot process with very general mutation, selection, and recombination operators and a function-valued dual process analogous to the block-counting process of the coalescent. Here, the function changes state as a jump process reminiscent of \eqref{eq:KingmanGenerator} due to genetic drift, selection, and recombination, while mutation contributes a deterministic component evolving the function continuously between jumps. Dualities in which mutation is either deterministic or absent can be used to compute some quantities of interest in the two models, but they are not the most general available. 
In this paper our purpose is different: it is to develop a coalescent dual for the Wright-Fisher diffusion in which mutation contributes to the \emph{random} evolution of the dual process. 
This type of duality is important because the dual process describes the posterior genealogical dynamics of a sample, conditional on the allelic configuration observed in the present day. This is precisely the process of interest when one wishes to perform statistical inference under a coalescent model given some sample of genetic variation taken from a contemporary population \citep[see][for an introduction]{ste:2007}. For example, a careful approximation of these dynamics provides a suitable proposal process in an importance sampling algorithm \citep[examples for multilocus models include][]{gri:mar:1996, fea:don:2001, lar:etal:2002, gri:etal:2008, lar:les:2008, jen:gri:2011, kam:etal:2016}. This duality is also important because it provides a way of obtaining an expression for the transition function of the underlying diffusion \citep{gri:1979:AAP11:310,don:tav:1987, eth:gri:1993}.

Dualities of this latter form have been developed for a number of models extending \eqref{eq:WFGenerator} and \eqref{eq:KingmanGenerator}. These include models of mutation \citep{gri:1980, don:tav:1987}, natural selection \citep{bar:etal:2000, fea:2002, ste:don:2003, eth:gri:2009}, and $\gL$-coalescent dynamics \citep{eth:etal:2010}, as well as dualities for the Moran model which is a prelimit of the corresponding diffusion \citep{eth:gri:2009, eth:etal:2010}. Hitherto, there has not been described a corresponding dual process for models incorporating both mutation and recombination (by which we mean homologous, meiotic, crossover). [The existence of one such process is implicit in \citet{fea:don:2001} and \citet{gri:etal:2008}, but there the focus was on inference rather than any description of the process.] The goal of this paper is to derive such a duality relationship from first principles: in particular, we identify a genealogical dual for the Wright-Fisher diffusion with recombination which is similar to the ancestral recombination graph ({\sc arg}) of \citet{gri:mar:1997}; the key differences being that here the lineages are typed, and jumps in the genealogical process are to be understood in an \emph{a posteriori} sense. We obtain results both for a finite-locus model with general mutation structure and for its limit with continuous breakpoint distribution and infinitely-many-sites mutation. Our key object of study is a generalisation of the generator $\sL$ defined in \eqref{eq:WFGenerator} and the duality identity \eqref{eq:duality}. As applications of our approach we recover systems of recursive equations for the sampling distribution of the models (usually obtained more toilsomely by direct coalescent arguments), and we also
obtain the first transition function expansion for a diffusion model incorporating recombination. Finally, we study the case of no mutation in further detail and develop an efficient method for computing the distribution of how ancestral genetic material is dispersed across the ancestors of a contemporary population (the so-called \emph{partitioning process}). Using duality, these distributions also yield fixation probabilities for haplotypes in multilocus models.

The paper is structured as follows. In \sref{sec:warmup} we illustrate our approach with a known example of a $K$-allele system at a single locus. We then extend this in \sref{sec:Llocus} to an $L$-locus model. 
In \sref{sec:transition} we apply these results to develop a series expansion for the transition function of the diffusion. In \sref{sec:continuous} we generalise the model further, to a continuous model of a gene in which mutation and recombination rates are modelled by a probability density function. In \sref{sec:nomutation} we return to the $L$-locus model and study in further detail the dual process of a Wright-Fisher diffusion without mutation, and \sref{sec:discussion} concludes with a brief discussion. 

\section{Warm up: $K$-alleles at one locus} 
\label{sec:warmup}
To illustrate the main idea and to clarify some notation, we first consider an extension of \eqref{eq:duality} to incorporate $K$-alleles with parent-independent mutation ({\sc pim}) at one locus. The key step is to make a judicious choice of duality function $F$ so that, when we apply to it the infinitesimal generator of the underlying diffusion as an operator on the first variable of $F$, we \emph{recognise} the resulting expression as the action of another generator acting on the second variable. Further applications of this idea can be found in \citet{eth:gri:1993}, \citet{bar:etal:2000}, and \citet{eth:gri:2009}.

Denote the finite type space of the locus by $E = \{1,\dots, K\} =: [K]$. The mutation model is specified by a rate parameter $\gq > 0$ and a distribution $(P_i)_{i\in E}$ over the type of a mutant offspring (independent of the parental allele). Within this framework, the Wright-Fisher diffusion $\bfX = (\bfX_t)_{t\geq 0}$ has state space
\begin{equation}
\label{eq:simplex}
\Delta_E = \left\{\bfx = (x_i)_{i\in E} \in [0,1]^E: \sum_{i\in E} x_i = 1\right\}
\end{equation}
and generator 
\begin{equation}
\label{eq:WFGeneratorK}
\sL f(\bfx) = \frac{1}{2}\sum_{i\in E}\sum_{j\in E}x_i(\gd_{ij} - x_j) \frac{\partial^2}{\partial x_i \partial x_j}f(\bfx) + \frac{\gq}{2}\sum_{i\in E}(P_i - x_i) \frac{\partial}{\partial x_i}f(\bfx),
\end{equation}
where $\gd_{ij}$ denotes the Kronecker delta, and $\sD(\sL) = C^2(\gD_E)$. Motivated by the choice of $F(x,n)$ we encountered above, let us evaluate $\sL F(\bfx,\bfn)$ for $F:\Delta_E \times \bbN^{|E|} \to \bbR$ defined by
\begin{equation}
\label{eq:F}
F(\bfx,\bfn) = \frac{1}{m(\bfn)}\prod_{i\in E} x_i^{n_i},
\end{equation}
for some $m:\bbN^{|E|} \to \bbR$ yet to be determined (here, $\bfn = (n_1,n_2,\dots,n_K) \in \bbN^K$ and $|Z|$ denotes the cardinality of a set $Z$). We find
\begin{multline}
\label{eq:LF}
\sL F(\cdot,\bfn)(\bfx) =\\ \sum_{i\in E}\frac{n_i(n_i + \gq P_i - 1)}{2}\frac{m(\bfn-\bfe_i)}{m(\bfn)}F(\bfx,\bfn - \bfe_i) - \frac{n(n+\gq - 1)}{2}F(\bfx,\bfn),
\end{multline}
where $\bfe_i = (\gd_{ij})_{j=1,\dots,K}$. 
This can be interpreted as the generator of a pure jump process evolving $\bfn$ on $\bbN^K$ if 
we can choose $m(\bfn)$ so that \eqref{eq:LF} is in the form
\begin{equation}
\label{eq:RateForm}
\sL F(\cdot,\bfn)(\bfx) = \sum_{\widehat{\bfn}} q(\bfn,\widehat{\bfn})[F(\bfx,\widehat{\bfn}) - F(\bfx,\bfn)],
\end{equation}
where $\bfQ = (q(\cdot,\cdot))$ is a rate matrix; that is, it has negative diagonal elements, nonnegative off-diagonal elements, and rows summing to 0. Now, take expectations in \eqref{eq:RateForm} with respect to the stationary distribution of $\bfX$ and use the identity
\begin{equation}
\label{eq:master}
\bbE[\sL F(\bfX_\infty,\bfn)] = 0
\end{equation}
to obtain
\begin{equation}
\label{eq:0Qv}
\bfzero = \bfQ\bfv,
\end{equation}
where $\bfv = (\bbE[F(\bfX_\infty,\cdot)])$ is a column vector. Here we generically use $\bfX_\infty$ to denote the process at stationarity. Equation \eqref{eq:0Qv} can be ensured if $\bfv$ has identical entries. In other words, we should choose $m(\bfn)$ in \eqref{eq:F} so that $\bbE[F(\bfX_\infty,\bfn)]$ is a constant (and without loss of generality, 1). Using that $\bfX_\infty \sim \text{Dirichlet}(\gq P_1,\gq P_2 ,\dots, \gq P_K)$ \citep{wri:1949}, we find by taking expectation in \eqref{eq:F} that we require
\begin{equation}
\label{eq:m}
m(\bfn) = \bbE\left[\prod_{i\in E}(\bfX_\infty)_i^{n_i}\right] = \frac{\prod_{i\in E} (\gq P_i)_{n_i}}{(\gq)_n},
\end{equation}
where, for $\phi \in \bbR_{\geq 0}$, $(\phi)_n := \phi(\phi+1)\dots (\phi+n-1)$ denotes the $n$th ascending factorial of $\phi$ (and $(\phi)_0 := 1$). Then \eqref{eq:LF} becomes 
\begin{equation*}
\sL F(\cdot,\bfn)(\bfx) = \sum_{i\in E}\frac{n_i(n + \gq - 1)}{2}F(\bfx,\bfn - \bfe_i) - \frac{n(n+\gq - 1)}{2}F(\bfx,\bfn),
\end{equation*}
which is of the required form. (Perhaps surprisingly, this result suggests that the generator of the dual process does not depend on the $P_i$. However, the $P_i$ do appear in the function $F$, which is not just an arbitrary function.)

In summary, the diffusion with generator \eqref{eq:WFGeneratorK} is dual to a pure death process on $\bbN^K$ with transition rate matrix
\begin{equation}
\label{eq:rate1locus}
q(\bfn,\widehat{\bfn}) = \frac{n+\gq-1}{2}\times\begin{cases}
\phantom{-}n_i & \text{if  }\widehat{\bfn} = \bfn-\bfe_{\bfi},\\
-n, & \text{if  }\widehat{\bfn} = \bfn,
\end{cases}
\end{equation}
and the duality function is
\begin{equation}
\label{eq:Fb}
F(\bfx,\bfn) = \frac{(\gq)_n}{\prod_{i\in E} (\gq P_i)_{n_i}}\prod_{i\in E} x_i^{n_i}.
\end{equation}
From \eqref{eq:rate1locus}, an interpretation of the dual process is as follows: at rate $n(n+\theta - 1)/2$, choose a gene to coalesce or mutate. At the chosen event, the gene involved is of type $i$ with probability $n_i/n$. It is well known that, under a {\sc pim} model, the posterior probability that any particular lineage was involved in the most recent event is independent of its type. At either type of event, the lineage involved is lost, which is reminiscent of coalescent simulation under the \emph{prior}: (only) under a {\sc pim} model, lineages undergoing mutation can be killed, so a simulated coalescent history becomes a random forest with each tree describing the genealogy of the sampled descendants of a single mutant.

Inspection of \eqref{eq:Fb} might lead one to suspect that the duality between the two processes is really about equivalence of sampling distributions. Let us unpick this further by plugging \eqref{eq:Fb} into the duality equation \eqref{eq:semigroupduality} and comparing the two sides. We contend that we have obtained two ways of addressing the following: What is the ratio of (i) the probability that a random sample of size $n$ results in an ordered allelic configuration which, when unordered, yields the vector $\bfn$, given that the population allele frequencies a time $t$ ago were $\bfx$; and (ii) the same probability without this extra information about the population? Using \eqref{eq:m}, the left of \eqref{eq:semigroupduality} is
\begin{equation}
\label{eq:duality-left}
\bbE[F(\bfX_t,\bfn)\mid \bfX_0 = \bfx] = \frac{\bbE\left[\prod_{i\in E}(\bfX_t)_i^{n_i} \mid \bfX_0 = \bfx\right]}{\bbE\left[\prod_{i\in E}(\bfX_\infty)_i^{n_i}\right]}.
\end{equation}
If our random sample is interpreted as an independent and identically distributed ({\sc iid}) set of $n$ draws with replacement at time $t$ from an infinite population evolving as a Wright-Fisher diffusion, then the quantity \eqref{eq:duality-left} is our claimed ratio of probabilities. Next, to interpret the right of \eqref{eq:semigroupduality}, we must be able to assign a prior on $\bfL_0$. The appropriate choice is of course the sampling distribution of the coalescent, which can be shown under a {\sc pim} model to be given by $m(\bfn)$ in \eqref{eq:m} (this is possible solely by coalescent arguments, without having to invoke the diffusion). Now, the right of \eqref{eq:semigroupduality} is
\[
\bbE[F(\bfx,\bfL_t)\mid \bfL_0 = \bfn] = \bbE\left[\frac{\prod_{i\in E} x_i^{(\bfL_t)_i}}{m(\bfL_t)}\Bigg|\,\bfL_0 = \bfn\right].
\]
The quantity inside the expectation is a ratio of: the probability of obtaining an ordered random sample with configuration $\bfL_t$ from a population in state $\bfx$ to the same probability under the coalescent prior. Two applications of Bayes' theorem then gives
\begin{multline}
\bbE\left[\frac{\prod_{i\in E} x_i^{(\bfL_t)_i}}{m(\bfL_t)}\Bigg|\,\bfL_0 = \bfn\right] = \bbE\left[\frac{\bbP(\bfX_0\in \ud\bfx\mid \bfL_t)}{\bbP(\bfX_0 \in \ud \bfx)}\Bigg|\, \bfL_0 = \bfn\right] \\
= \frac{\bbP(\bfX_0\in \ud\bfx\mid \bfL_0 = \bfn)}{\bbP(\bfX_0 \in \ud \bfx)} 
= \frac{\bbP(\bfL_0 = \bfn \mid \bfX_0 = \bfx)}{\bbP(\bfL_0 = \bfn)},\label{eq:duality-right}
\end{multline}
which is again the claimed ratio (recalling that time 0 is different for $\bfL$ and $\bfX$). 
The right of \eqref{eq:semigroupduality} is therefore a ratio of \emph{coalescent} sampling probabilities. The numerator is the probability for a random sample with configuration $\bfn$ given that the lineages ancestral to this sample a time $t$ ago were typed by {\sc iid} sampling from a population in state $\bfx$, while the denominator is the same probability without this additional information. 
Under this interpretation, the duality function \eqref{eq:Fb} is also a ratio of sampling distributions, now without any offset of time:
\[
F(\bfx,\bfn) = \frac{\bbP(\bfL_0 = \bfn\mid \bfX_t = \bfx)}{\bbP(\bfL_0 = \bfn)}.
\]


\section{An $L$-locus model}
\label{sec:Llocus}
In this section we extend the above ideas to a multilocus model in which recombination can occur between each locus. We allow for more general mutation models than in \sref{sec:warmup}, though for convenience we continue to assume that the type space at each locus is finite. We first introduce some notation. Suppose a haplotype is determined by the alleles at each of $L$ loci. The set of possible alleles at locus $l$ is denoted $E_l$, so that the set of all possible haplotypes is $E = \bigtimes_{l=1}^L E_l$. The frequency of haplotype $\bfi =(i_1,\ldots,i_L) \in E$ will be denoted by $x_\bfi$. The mutation parameter at locus $l$ is $\gq_l$ and mutation occurs at that locus according to a transition matrix $\bfP^{(l)} = (P^{(l)}_{ij})_{i,j \in E_l}$; in other words, when a mutation occurs to a haplotype with allele $i$ at locus $l$, its offspring has allele $j$ at that locus with probability $P^{(l)}_{ij}$. We will denote the resulting haplotype by $\m{\bfi}{l}{j} := (i_1,\dots,i_{l-1},j,i_{l+1},\dots,i_L)$. 
Mutation occurs independently at each locus, so we may define mutation parameters across all loci as:
\begin{equation}
\label{eq:FiniteSites}
\gq = \sum_{l=1}^L \gq_l, \quad
\bfP = \sum_{l=1}^L \frac{\gq_l}{\gq} \bfI_{|E_1|} \otimes \dots \otimes \bfI_{|E_{l-1}|}  \otimes \bfP^{(l)} \otimes \bfI_{|E_{l+1}|} \otimes \dots \otimes \bfI_{|E_L|},
\end{equation}
where $\otimes$ denotes outer product, $\bfI_d$ is the $d \times d$ identity matrix, and $\bfP^{(l)}$ appears in the $l$th term in the product. 
Notice that if mutation is parent-independent at each locus (so $P_{ij}^{(l)} = P_j^{(l)}$ for each $l\in [L]$, $i,j \in E_l$), then the allele frequencies at each locus, $(X^{\{l\}}_{i_l})_{i_l \in E_l}$ with $X^{\{l\}}_{i_l} = \sum_{\bfj\in E :\, j_l = i_l} X_\bfj$, evolve marginally according to the one-locus model of \sref{sec:warmup}. 
For each $l=1,\ldots,L-1$, the rate of recombination between locus $l$ and $l+1$ is parametrised by $\gr_l$, and we let $\gr = \sum_{l=1}^{L-1} \gr_l$.

For a nonempty subset $A \subseteq [L]$, denote the projection of $E$ onto the co-ordinates in $A$ by $E_A$, i.e.~$E_A = \bigtimes_{l\in A}E_l$. Denote the marginal frequency of the alleles $\bfi \in E_A$ by
\begin{equation*}
x_{\bfi}^{A} = \sum_{\bfj\in E:\,\bfj|_A = \bfi} x_\bfj.
\end{equation*}
Sometimes we will also write $x_{\bfi}^A$ for $\bfi \in E_B$ and $B \supset A$, by which it is implied that we mean
\begin{equation}
\label{eq:marginalbis}
x_{\bfi|_A}^A= \sum_{\bfj\in E:\,\bfj|_A = \bfi|_A } x_\bfj.
\end{equation}
Finally, for $A \subseteq [L]$ we also define the sets
\begin{align*}
\Al &= A \cap \{1,\ldots ,l\}, &
\Ag &= A \cap \{l+1,\ldots,L\}.
\end{align*}
With this new definition for $E$, the multilocus Wright-Fisher diffusion process with recombination has state space $\Delta_E$ as in \eqref{eq:simplex}. Its generator is given by
\begin{multline}
\label{eq:WFGeneratorR}
	\sL = \frac{1}{2}\sum_{\bfi\in E}\left[\sum_{\bfj\in E} x_{\bfi}(\gd_{\bfi\bfj}-x_{\bfj})\frac{\partial}{\partial x_{\bfj}}\right.\\
	+ \left. \sum_{l=1}^L \gq_l \left[\sum_{j\in E_l}P^{(l)}_{ji_l}x_{\m{\bfi}{l}{j}} - x_{\bfi}\right] + \sum_{l=1}^{L-1} \gr_l(x_{\bfi}^{\Ll}x_{\bfi}^{\Lg} - x_{\bfi})\right]\frac{\partial}{\partial x_{\bfi}}
\end{multline}
and $\sD(\sL) = C^2(\gD_E)$.
\subsection{An `unreduced' dual}
\label{sec:unreduced}
To obtain the dual process of \eqref{eq:WFGeneratorR}, we follow the strategy outlined in \sref{sec:warmup}. First consider the test function corresponding to the unordered sampling distribution of $\bfn$:
\begin{equation}
\label{eq:Q}
\Q(\bfx,\bfn) = \binom{n}{\bfn}\prod_{\bfi\in E}x_\bfi^{n_\bfi},
\end{equation}
where $\binom{n}{\bfn} = n!/\prod_{\bfi\in E}n_{\bfi}!$ is the multinomial coefficient. We know from \sref{sec:warmup} that, as a function of $\bfx$, our duality function will be proportional to $\Q(\bfx, \bfn)$. 
In fact, rather than consider $\Q(\bfx,\bfn)$ directly, we can work with the probability generating function ({\sc pgf})
\begin{equation}
	\label{eq:generatorGs}
G_n(\bfs;\bfx) = \sum_{\bfn\in \nabla_{E,n}}\left[\prod_{\bfi\in E}s_\bfi^{n_\bfi} \right]\Q(\bfx,\bfn) = \left[\sum_{\bfi\in E}s_\bfi x_\bfi\right]^n,
\end{equation}
where $\bfs = (s_\bfi)_{\bfi\in E}$ and
\[
\nabla_{E,n} = \left\{\bfn = (n_\bfi)_{\bfi\in E} \in \bbN^{|E|}: \sum_{\bfi\in E} n_\bfi = n\right\},
\]
and then recover $\Q(\bfx,\bfn)$ from this later. (Here and throughout, define $\Q(\bfx,\bfn) = 0$ if $\bfx \not\in \gD_E$ or $\bfn \not\in \nabla_{E,n}$ for any $n$.) For other examples of the use of generating functions in the context of population genetics models with recombination, see \citet{gri:1981}, \citet{eth:gri:1990:JMB}, \citet{gri:1991}, and \citet{loh:etal:2011, loh:etal:2016}.

A simple calculation yields
\begin{multline}
	\label{eq:generatorGs2}
	\sL G_n(\bfs; \bfx) = \sum_{\bfi\in E}\left[\binom{n}{2}s_\bfi^2 x_\bfi G_{n-2}(\bfs;\bfx)
	+ \sum_{l=1}^L \frac{\gq_ln}{2} \sum_{j\in E_l} s_\bfi P^{(l)}_{ji_l} x_{\m{\bfi}{l}{j}}G_{n-1}(\bfs;\bfx)\right.\\
	\left. {}+ \sum_{l=1}^{L-1}\frac{\gr_ln}{2} s_\bfi x_{\bfi}^{\Ll} x_{\bfi}^{\Lg} G_{n-1}(\bfs;\bfx)\right] - \frac{n(n-1+\gq+\gr)}{2}G_n(\bfs;\bfx).
\end{multline}
The remainder of the strategy would be (i) to extract the an equation for $\sL S(\bfx,\bfn)$ from \eqref{eq:generatorGs2}, (ii) rearrange this equation into the required dual form, and (iii) read off a rate matrix for the dual process. However, we can see from \eqref{eq:Q}--\eqref{eq:generatorGs2} that no distinction has been made between loci that are ancestral and those that are non-ancestral with respect to an `initial' (present-day) sample. Consequently, the dual process would track both types of loci. This is the posterior analogue of the {\sc arg} of \citet{gri:mar:1997}, in which the total number of lineages can grow unboundedly 
backwards in time. It would be preferable to construct an analogue of the `reduced' version of the {\sc arg} in which only lineages ancestral to the initial sample are traced back in time \citep[see, e.g.][]{hud:1983, gol:1984, eth:gri:1990:JMB, gri:1991, gri:etal:2008}. We therefore move straight to the following subsection in which we construct a correspondingly reduced version of the dual process.

\subsection{A `reduced' dual}
\label{sec:reduced}
The state space for our reduced dual process will be
\[
\Xi_{E,n} = \left\{\bfn = (n_{\bfi}^A)_{\emptyset\neq A\subseteq [L],\bfi\in E_A}:  n_{\bfi}^A \in \bbN,\sum_{\emptyset\neq A\subseteq [L]}\sum_{\bfi \in E_A} n_{\bfi}^A = n\right\}.
\]
The set $A$ records those loci at which the haplotype $\bfi \in E_A$ is ancestral to an initial (present-day) sample, and the alleles at only those loci are recorded. The notation $n_{\bfi}^A$ is then the number of times the haplotype $\bfi$ is observed, and we will also let $n^A = \sum_{\bfi \in E_A}n_\bfi^A$. 
By analogy with the previous subsection, we define the test function
\begin{equation}
\label{eq:Q2}
\widetilde{\Q}(\bfx,\bfn) = \binom{n}{\bfn}\prod_{\emptyset\neq A\subseteq [L]}\prod_{\bfi\in E_A} (x_{\bfi}^{A})^{n_{\bfi}^A}.
\end{equation}
for $\bfx \in \gD_E$, $\bfn \in \cup_{n=1}^\infty \Xi_{E,n}$ (and $\widetilde{\Q}(\bfx,\bfn) = 0$ otherwise); and the generating function
\begin{multline}
	\label{eq:generatorGt}
\widetilde{G}_n(\bft;\bfx) = \sum_{\bfn\in\Xi_{E,n}}\prod_{\emptyset \neq A \subseteq [L]}\prod_{\bfi\in E_A} (t_{\bfi}^
{A})^{{n}_{\bfi}^A} \widetilde{\Q}(\bfx,\bfn)\\ = \left(\sum_{\emptyset \neq A\subseteq [L]}\sum_{\bfi\in E_A} t_{\bfi}^Ax_{\bfi}^A\right)^n = \left[\sum_{\bfj\in E}\left(\sum_{\emptyset \neq A\subseteq [L]}t_{\bfj|_A}^A\right)x_\bfj\right]^n,
\end{multline}
with dummy variables $\bft = (t_{\bfi}^A)_{\emptyset \neq A \subseteq [L],\bfi \in E_A}$, where the last equality follows from \eqref{eq:marginalbis} and reordering the summations.


Now our use of generating functions pays off. Comparing the right-hand expression in \eqref{eq:generatorGt} with \eqref{eq:generatorGs} shows that to evaluate $\sL \widetilde{G}_n(\bft;\bfx)$ we simply need to apply the mapping
\[
s_\bfi \mapsto \sum_{\emptyset\neq A\subseteq [L]} t_{\bfi|_A}^A
\]
in \eqref{eq:generatorGs2}. 
After some rearrangement we obtain
\begin{align}
	\label{eq:generatorGt2}
	\sL \widetilde{G}_n(\bft; \bfx) = {}& \sum_{\emptyset \neq A \subseteq [L]}\left[\binom{n}{2}\sum_{\emptyset\neq B\subseteq [L]}\sum_{\bfi\in E_{A\cup B}} t_{\bfi}^At_{\bfi}^B x_{\bfi}^{A\cup B} \widetilde{G}_{n-2}(\bft;\bfx) \right.\nonumber\\
	{}& \phantom{\sum_{\emptyset \neq A \subseteq [L]}} + \sum_{\bfi\in E_A} \sum_{l\in A} \frac{\gq_ln}{2} t_{\bfi}^A\sum_{j\in E_l} P^{(l)}_{ji_l} x_{\m{\bfi}{l}{j}}^A\widetilde{G}_{n-1}(\bft;\bfx) \nonumber\\
	 {}& \phantom{\sum_{\emptyset \neq A \subseteq [L]}} \left. + \sum_{\bfi\in E_A} \sum_{l=\min A}^{\max A -1} \frac{\gr_l n}{2} t_{\bfi}^A x_{\bfi}^{\Al}x_{\bfi}^{\Ag} \widetilde{G}_{n-1}(\bft;\bfx)\right] \nonumber\\
{}& -\left[\frac{n(n-1)}{2}\widetilde{G}_n(\bft;\bfx) +\sum_{l=1}^L \frac{n\gq_l}{2} \sum_{\substack{A\subseteq [L]:\\ l\in A}}\,\sum_{\bfi\in E_A} t_{\bfi}^Ax_{\bfi}^A\widetilde{G}_{n-1}(\bft;\bfx) \right.\nonumber\\
	{}& \phantom{\sum_{\emptyset \neq A \subseteq [L]}} \left. +\sum_{l=1}^{L-1}\frac{n\gr_l}{2} \sum_{\substack{A\subseteq [L]:\\  \Al \neq \emptyset\neq\Ag}}\sum_{\bfi\in E_A}t_{\bfi}^Ax_{\bfi}^A\widetilde{G}_{n-1}(\bft;\bfx)\right].
\end{align}
Now we can continue the strategy outlined in the previous subsection. Noting that
\[
\sL \widetilde{G}_n(\bft;\bfx) = \sum_{\bfn\in\Xi_{E,n}}\prod_{\emptyset \neq A \subseteq [L]}\prod_{\bfi\in E_A} (t_{\bfi}^
{A})^{{n}_{\bfi}^A} \sL\widetilde{\Q}(\bfx,\bfn),
\]
we can compare coefficients of $\prod_{\emptyset\neq A\subseteq [L]}\prod_{\bfi\in E_A}(t_{\bfi}^{A})^{n_{\bfi}^A}$ in \eqref{eq:generatorGt2} to obtain
{\allowdisplaybreaks
\begin{align}
	\sL \widetilde{\Q}(\bfx,\bfn) = {}& \sum_{\emptyset \neq A \subseteq [L]}\left[\binom{n}{2}\sum_{\emptyset\neq B\subseteq [L]}\sum_{\bfi\in E_{A\cup B}} x_{\bfi}^{A\cup B} \widetilde{\Q}(\bfx,\bfn-\bfe_\bfi^A - \bfe_\bfi^B) \right.\nonumber\\
	{}& \phantom{\sum_{\emptyset \neq A \subseteq [L]}} + \sum_{\bfi\in E_A} \sum_{l\in A} \frac{\gq_ln}{2} \sum_{j\in E_l} P^{(l)}_{ji_l} x_{\m{\bfi}{l}{j}}^A\widetilde{\Q}(\bfx,\bfn-\bfe_\bfi^A) \nonumber\\
	 {}& \phantom{\sum_{\emptyset \neq A \subseteq [L]}} \left. + \sum_{\bfi\in E_A} \sum_{l=\min A}^{\max A -1} \frac{\gr_l n}{2} x_{\bfi}^{\Al}x_{\bfi}^{\Ag} \widetilde{\Q}(\bfx,\bfn-\bfe_\bfi^A)\right] \nonumber\\
	{}& - \left[\frac{n(n-1)}{2}+\sum_{l=1}^L \frac{\gq_l}{2} \sum_{\substack{A\subseteq [L]:\\ l\in A}}n^A +\sum_{l=1}^{L-1}\frac{\gr_l}{2} \sum_{\substack{A\subseteq [L]:\\  \Al \neq \emptyset\neq\Ag}}n^A\right]\widetilde{\Q}(\bfx,\bfn). \label{eq:generatorQ}
\end{align}} To manipulate this into dual form, we further rearrange the right-hand side in order to remove the explicit instances of $\bfx$ outside of $\Q(\bfx,\cdot)$. Using \eqref{eq:ex1}--\eqref{eq:ex3} of \ref{app:identities} together with \eqref{eq:generatorQ}, we obtain
\begin{align}
	\label{eq:generatorQ3}
	\lefteqn{\sL \widetilde{\Q}(\bfx, \bfn) =}\hspace{50pt}&\notag\\
	{}& \frac{1}{2}\sum_{\emptyset\neq A\subseteq [L]}\left[\sum_{\emptyset\neq B\subseteq [L]}\sum_{\bfi\in E_{A\cup B}} n(n_{\bfi}^{A\cup B} + 1 - \gd_{A,A\cup B} - \gd_{B,A\cup B}) \right.\notag\\
&	\specialcell{\hfill\times\widetilde{\Q}(\bfx,\bfn-\bfe_{\bfi}^A-\bfe_{\bfi}^B+\bfe_{\bfi}^{A\cup B})}\nonumber\\
	{}& {}+\sum_{\bfi\in E_A} \sum_{l\in A} \gq_l \sum_{j\in E_l} P^{(l)}_{ji_l} (n_{\m{\bfi}{l}{j}}^A+1-\gd_{i_lj})\widetilde{\Q}(\bfx,\bfn-\bfe_{\bfi}^A+\bfe^A_{\m{\bfi}{l}{j}}) \nonumber\\
	 {}& \left. {}+ \sum_{\bfi\in E_A}\sum_{l=\min A}^{\max A-1}\gr_l \frac{(n_{\bfi}^{\Al}+1)(n_{\bfi}^{\Ag}+1)}{n+1}\widetilde{\Q}(\bfx,\bfn-\bfe_{\bfi}^A+\bfe_{\bfi}^{\Al}+\bfe_{\bfi}^{\Ag})\right] \nonumber\\
	{}& - \left[\frac{n(n-1)}{2}+\sum_{l=1}^L \frac{\gq_l}{2} \sum_{\substack{A\subseteq [L]:\\ l\in A}}n^A +\sum_{l=1}^{L-1}\frac{\gr_l}{2} \sum_{\substack{A\subseteq [L]:\\  \Al \neq \emptyset\neq\Ag}}n^A\right]\widetilde{\Q}(\bfx,\bfn).
\end{align}
If we divide \eqref{eq:generatorQ3} by $\bbE[\widetilde{\Q}(\bfX_\infty,\bfn)]$ then, after a little rearrangement, we have succeeded in writing $\sL F(\bfx,\bfn)$ in the form of \eqref{eq:RateForm} for the duality function
\begin{equation}
\label{eq:F3}
\widetilde{F}(\bfx,\bfn) = \frac{\widetilde{\Q}(\bfx,\bfn)}{\bbE[\widetilde{\Q}(\bfX_\infty,\bfn)]},
\end{equation}
from which we can read off the rate matrix for the dual process on $\cup_{n=1}^\infty \Xi_{E,n}$. We have therefore shown the following.
\begin{thm}
\label{thm:reducedFiniteSites}
Let
\begin{equation}
\label{eq:moments}
\widetilde{m}(\bfn) = \bbE\left[\prod_{\emptyset\neq A\subseteq [L]}\prod_{\bfi\in E_A}(X_{\bfi}^{A})^{{n}_{\bfi}^A}\right]
\end{equation}
(for $\bfn \in \cup_{n=1}^\infty \Xi_{E,n}$ and 0 otherwise), where expectation is taken with respect to the stationary distribution of $\bfX$. The Wright-Fisher diffusion $\bfX = (\bfX_t)_{t\geq 0}$ on $\gD_E$ with generator \eqref{eq:WFGeneratorR} is dual to a pure jump process $\widetilde{\bfL} = (\widetilde{\bfL}_t)_{t\geq 0}$ on $\cup_{n=1}^\infty \Xi_{E,n}$ with transitions given by the following description.
\begin{description}
\item[~Coalescence.] For each nonempty $A,B\subseteq [L]$ and each $\bfi\in E_{A\cup B}$, the process jumps to $\bfn-\bfe_\bfi^A - \bfe_\bfi^B + \bfe_\bfi^{A\cup B}$ at rate
\[
\frac{1}{2}\frac{\widetilde{m}(\bfn-\bfe_\bfi^A - \bfe_\bfi^B + \bfe_\bfi^{A\cup B})}{\widetilde{m}(\bfn)}n_\bfi^A(n_\bfi^B - \gd_{AB}).
\]
\item[~Mutation.] For each nonempty $A \subseteq [L]$, $l \in A$, $\bfi \in E_A$, and $j \in E_l$, the process jumps to $\bfn-\bfe_\bfi^A + \bfe_{\m{\bfi}{l}{j}}^{A}$ at rate
\[
\frac{1}{2}\frac{\widetilde{m}(\bfn-\bfe_\bfi^A + \bfe_{\m{\bfi}{l}{j}}^{A})}{\widetilde{m}(\bfn)}n_\bfi^A\gq_l P_{ji_l}^{(l)}.
\]
\item[~Recombination.] For each nonempty $A\subseteq [L]$, $\bfi \in E_A$, and $l= \min A,\dots$, $\max A -1$, the process jumps to $\bfn - \bfe_\bfi^A +\bfe_\bfi^{\Al} + \bfe_\bfi^{\Ag}$ at rate
\[
\frac{1}{2}\frac{\widetilde{m}(\bfn - \bfe_\bfi^A +\bfe_\bfi^{\Al} + \bfe_\bfi^{\Ag})}{\widetilde{m}(\bfn)}n_\bfi^A\gr_l.
\]
\end{description}
The duality function relating the two processes is $\widetilde{F}(\bfx,\bfn)$, given by \eqref{eq:F3} and \eqref{eq:Q2}.
\end{thm}
\begin{rem}
\label{rem:Q}
It is straightforward, though notationally cumbersome, to construct $\bfQ$ from the description given in \thmref{thm:reducedFiniteSites}. A given transition rate $q(\bfn,\widehat{\bfn})$ is obtained by summing over the rates in \thmref{thm:reducedFiniteSites} that correspond to a particular destination state $\widehat{\bfn}$. 
\end{rem}
\begin{cor}
\label{cor:transient}
The transient sampling distributions of $\bfX$ and $\widetilde{\bfL}$ are related by
\[
\bbE[\widetilde{\Q}(\bfX_t,\bfn)\mid \bfX_0 = \bfx] = \bbE\left[\frac{\bbE[\widetilde{\Q}(\bfX_\infty,\bfn)]}{\bbE[\widetilde{\Q}(\bfX_\infty,\widetilde{\bfL}_t)\mid \widetilde{\bfL}_t]}\widetilde{\Q}(\bfx,\widetilde{\bfL}_t)\mid \widetilde{\bfL}_0 = \bfn\right].
\]
\end{cor}
\begin{proof}
This follows immediately from the duality equation
\[
\bbE\left[\widetilde{F}(\bfX_t,\bfn)\mid \bfX_0 = \bfx\right] = \bbE\left[\widetilde{F}(\bfx,\widetilde{\bfL}_t)\mid \widetilde{\bfL}_0 = \bfn\right]
\]
and \eqref{eq:F3}.
\end{proof}
It is possible to provide a genealogical interpretation of \thmref{thm:reducedFiniteSites} in a spirit similar to that given in \sref{sec:warmup}, the main differences being that here we account for recombination between multiple loci and construct the dual process so that it tracks only lineages ancestral to the initial sample. In summary, the duality function \eqref{eq:F3} is proportional to the ordered sampling distribution $\prod_{\emptyset\neq A\subseteq [L]}\prod_{\bfi\in E_A}(X_{\bfi}^{A})^{{n}_{\bfi}^A}$ of a haplotype configuration $\bfn$, when sampling is performed {\sc iid} from a population with haplotype frequencies $\bfx$. In this interpretation, the set of loci $A$ at which a sampled haplotype is actually observed 
is nonrandom. The normalisation constant of \eqref{eq:F3} is then $\widetilde{m}(\bfn)$, the sampling distribution for $\bfn$ when the population haplotype frequencies are at stationarity; this ensures both that equation \eqref{eq:generatorQ3} can easily be identified as the generator of a process acting on $\bfn$, and that the duality equation \eqref{eq:semigroupduality} has a straightforward interpretation as two ways of looking at (a ratio of) sampling probabilities. In this duality equation it is necessary to consider the genealogy of the present-day configuration $\bfn$ conditioned on the past state of the population, which gives rise to the posterior coalescent dynamics captured by the process $\widetilde{\bfL}$ described in \thmref{thm:reducedFiniteSites}. The ratios of terms in $\widetilde{m}(\cdot)$ in the transition rates of $\widetilde{\bfL}$ appear naturally as a time-reversal of the coalescent process. 
 
Of course, a major complication of the dual process here compared to that of \sref{sec:warmup} is that there is no closed-form expression for the stationary moments $\widetilde{m}(\bfn)$ of $\bfX$ [eq.\ \eqref{eq:moments}]. However, we can show that they satisfy a simple linear system.
\begin{prop}
\label{prop:systemFiniteSites}
For $\bfn \in \Xi_{E,n}$, 
the stationary 
moments $\widetilde{m}(\bfn)$ of \eqref{eq:moments} satisfy
\begin{align}
	\lefteqn{\left[n(n-1)+\sum_{l=1}^L \gq_l \sum_{\substack{A\subseteq [L]:\\ l\in A}}n^A +\sum_{l=1}^{L-1}\gr_l \sum_{\substack{A\subseteq [L]:\\  \Al \neq \emptyset\neq\Ag}}n^A\right]\widetilde{m}(\bfn) =}\hspace{40pt}&\notag\\
	{}& \sum_{\emptyset\neq A\subseteq [L]}\left[\sum_{\emptyset\neq B\subseteq [L]}\sum_{\bfi\in E_{A\cup B}} n_\bfi^A(n_\bfi^B - \gd_{AB}) \widetilde{m}(\bfn-\bfe_{\bfi}^A-\bfe_{\bfi}^B+\bfe_{\bfi}^{A\cup B})\right.\nonumber\\
	{}& {}+\sum_{\bfi\in E_A} \sum_{l\in A} \gq_l \sum_{j\in E_A} P^{(l)}_{ji_l} n_\bfi^A\widetilde{m}(\bfn-\bfe_{\bfi}^A+\bfe_{\m{\bfi}{l}{j}}^A) \nonumber\\
	 {}& \left. {}+ \sum_{\bfi\in E_A}\sum_{l=\min A}^{\max A-1}\gr_l n_\bfi^A\widetilde{m}(\bfn-\bfe_{\bfi}^A+\bfe_{\bfi}^{\Al}+\bfe_{\bfi}^{\Ag})\right]. \label{eq:SamplingDistributionFiniteSites}
\end{align}
A boundary condition is
\[
\widetilde{m}(\bfe_{i_1}^{[1]} + \bfe_{i_2}^{[2]} + \dots +\bfe_{i_L}^{[L]}) = \prod_{l=1}^L \mu^{(l)}_{i_l}, \qquad i_l \in E_l,
\]
where $\bfmath{\mu}^{(l)} = (\mu^{(l)}_1,\mu^{(l)}_2,\dots,\mu^{(l)}_{|E_l|})$ is the stationary distribution of $\bfP^{(l)}$.
\end{prop}
\begin{proof}
Take expectation with respect to the stationary distribution of $\bfX$ in \eqref{eq:generatorQ3} and apply the identity $\bbE[\sL\widetilde{\Q}(\bfX_\infty,\bfn)] = 0$ to get \eqref{eq:SamplingDistributionFiniteSites}. The boundary condition follows by the argument of \citet[Theorem 1]{fea:2003:JAP}.
\end{proof}
The advantage of a \emph{reduced} dual is now apparent. If we define the \emph{degree} of $\bfn$ by
\[
\text{degree}(\bfn) = \sum_{\emptyset \neq A\subseteq [L]}|A|n^A,
\]
the total length of all ancestral material in the sample, then the system \eqref{eq:SamplingDistributionFiniteSites} is \emph{closed} in the sense that terms on the right of \eqref{eq:SamplingDistributionFiniteSites} have degree less than or equal to that of $\bfn$, and so it can in principle be solved (e.g.\ by matrix inversion). 
The process $\widetilde{\bfL}$ evolves on a \emph{finite} state space. This is not true of the unreduced dual.


Recursive systems similar to \eqref{eq:SamplingDistributionFiniteSites} have been studied by \citet{gri:1981}, 
\citet{gol:1984}, \citet{eth:gri:1990:JMB}, \citet{lar:etal:2002}, \citet{fea:2003:JAP}, \citet{gri:etal:2008}, \citet{jen:son:2009:G}, \citet{lar:les:2008}, and \citet{jen:gri:2011}, among others. With the exception of \citet{lar:les:2008}, whose eq.\ (1) is equal to \eqref{eq:SamplingDistributionFiniteSites} up to a combinatorial factor, typically these studies focus on special cases such as two loci or parameters uniform across loci. It is common in studying systems of this form to derive them by a probabilistic argument; in particular, by describing the associated coalescent process and partitioning on each of the most recent possible events going back in time. This approach can be combinatorially involving, and we emphasise the cleanliness of the method taken in this paper: once we have the generator \eqref{eq:WFGeneratorR}, the rest follows mechanistically.

\subsection{A closed-form solution}
One special case of the above model permits a closed-form solution:
mutation within each locus is parent-independent ($P^{(l)}_{ij} = P^{(l)}_j$ for each $l\in [L]$, $i,j \in E_l$), and $\gr_l = \infty$ for each $l\in [L]$. 
Then each locus evolves independently, 
and the dual process is projected onto the subspace $\Xi^\infty_{E,n} \subseteq \Xi_{E,n}$ given by
\[
\Xi^\infty_{E,n} = \left\{ \bfn \in \Xi_{E,n}:  n^A =0\;\forall A \notin \left\{\{1\},\{2\},\dots,\{L\}\right\}\right\};
\]
that is, one in which any haplotype is ancestral at precisely one locus. The projection is achieved by mapping a haplotype $\bfi \in E_A$ with $A = \{a_1,\dots,a_{|A|}\}$ to $|A|$ different haplotypes of type $i_1 \in E_{a_1}$, $i_2 \in E_{a_2},\dots,i_{|A|} \in E_{a_{|A|}}$; recombination instantaneously breaks apart each locus. The generator for this model is a sum of $L$ generators acting independently on each locus \citep[see][for further details]{eth:gri:1990:AAP}, from which we can write down the transition rates of the dual process on $\cup_{m=1}^{Ln}\Xi^\infty_{E,n}$:
\[
\widetilde{q}(\bfn,\widehat{\bfn}) = \frac{1}{2}\times \begin{cases}
n_{i}^{\{l\}}(n^{\{l\}} + \gq_l - 1) & \text{if  }\widehat{\bfn} = \bfn-\bfe_{i}^{\{l\}}\\
& \text{ where  }l\in[L],i \in E_{l},\\
\ds-\sum_{l=1}^L n^{\{l\}}(n^{\{l\}}+\gq_l-1) & \text{if  }\widehat{\bfn} = \bfn.
\end{cases}
\]
The duality function in this case is
\[
\widetilde{F}(\bfx,\bfn) = \prod_{l=1}^L\left[\frac{(\gq_l)_{n^{\{l\}}}}{\prod_{i\in E_l} (\gq_l P^{(l)}_i)_{n_i^{\{l\}}}}\prod_{i\in E_l} (x_i^{\{l\}})^{n^{\{l\}}_i}\right],
\]
which is simply the product of $L$ copies of the one-locus duality function encountered in \sref{sec:warmup}, as it must be under free recombination.
\section{A transition function expansion}
\label{sec:transition}
Duality can be used to obtain an expression for the transition function of the Wright-Fisher diffusion. Here we tackle the diffusion with generator \eqref{eq:WFGeneratorR}, whose transition density with respect to Lebesgue measure, after evolving from $\bfx$ for a time $t$, we denote by $f(\bfx,\cdot; t)$, and whose stationary distribution we denote by $\pi(\cdot)$. 
To our knowledge this is the first time an expression for the transition function of a Wright-Fisher diffusion has incorporated recombination.

For simplicity we restrict our attention to `completely specified' samples: those for which $n_\bfi^A = 0$ if $A \neq [L]$, and we write $n_\bfi$ for $n_\bfi^{[L]}$, and so on. Then the sampling distribution of $\bfn$ can be written in the simpler form of \eqref{eq:Q}. Our result will be expressed in terms of the transitions of the dual process, which we denote $p_{\bfn\bfl}(t) := \bbP(\widetilde\bfL_t = \bfl\mid \widetilde\bfL_0 = \bfn)$. In particular, we let $n \to \infty$ in such a way that $\bfn/n \to \bfy \in \gD_E$ \citep[this idea is formalised by][p125]{bar:etal:2000} and write
\begin{equation}
\label{eq:pylt}
p_{\bfy\bfl}(t) := \lim_{n\to\infty: \frac{\bfn}{n}\to \bfy} p_{\bfn\bfl}(t).
\end{equation}
The existence of this limit ensures that our typed, reduced, coalescent process $\widetilde\bfL$ can be initiated from infinitely many lineages.

\begin{thm} Suppose that \eqref{eq:pylt} defines a probability distribution on $\bigcup_{n=1}^\infty \Xi_{E,n}$ for each $t > 0$, $\bfy \in \gD_E$. Then the transition density function of the Wright-Fisher diffusion with generator \eqref{eq:WFGeneratorR} is given by
\begin{equation}
\label{eq:transition}
f(\bfx,\bfy; t) = \pi(\bfy) \sum_{\bfl\in\bigcup_{l \in \bbN}\Xi_{E,l}} \frac{p_{\bfy\bfl}(t)}{\widetilde{m}(\bfl)}\ds\prod_{\emptyset \neq A \subseteq [L]} \prod_{\bfi \in E_A} (x_\bfi^A)^{l_\bfi^A},
\end{equation}
with $\widetilde{m}(\cdot)$ as in \eqref{eq:moments}.
\end{thm}
\begin{proof}
The proof is similar to the rigorous treatment given in \citet{bar:etal:2000} and so we give only a summary. Corollary \ref{cor:transient} 
easily leads to 
\begin{multline}
\bbE\left[\Q(\bfX_t,\bfn)
|\, \bfX_0 = \bfx\right] = \binom{n}{\bfn}m(\bfn)\bbE\left[\frac{\ds \prod_{\emptyset\neq A \subseteq [L]} \prod_{\bfi \in E_A} (x_\bfi^A)^{\widetilde L_\bfi^A(t)}}{\widetilde{m}(\widetilde\bfL_t)}\Bigg|\, \widetilde\bfL_0 = \bfn\right]\\
= \binom{n}{\bfn}m(\bfn)\sum_{\bfl\in\bigcup_{l \leq Ln}\Xi_{E,l}} \frac{p_{\bfn\bfl}(t)}{\widetilde{m}(\bfl)}\ds\prod_{\emptyset \neq A \subseteq [L]} \prod_{\bfi \in E_A} (x_\bfi^A)^{l_\bfi^A}.\label{eq:inversion}
\end{multline}
Our aim is to let $n \to \infty$ and $\bfn/n \to \bfy$ in \eqref{eq:inversion}. Letting $n \to \infty$ on the left-hand side is tantamount to identifying a distribution from its moments. \citet{eth:gri:2009} note that this is an application of `sample inversion': for a continuous function $u: \gD_E \to \bbR$ and a random sample $\bfN \sim \text{Multinomial}(n,\bfz)$,
\[
\bbE\left[u\left(\frac{\bfN}{n}\right)\right] = \sum_{\bfk \in \nabla_{E,n}} u\left(\frac{\bfk}{n}\right)\Q(\bfz,\bfk)
\to u(\bfz), \qquad n\to\infty,
\]
uniformly in $\bfz \in \gD_E$. 
To use this result we multiply both sides of \eqref{eq:inversion} by $u(\bfn/n)$. If $u$ is a function such that $u(\bfk/n) = 0$ if $\bfk \neq \bfn$ then the left-hand side of the resulting equation is
\begin{align*}
\bbE\left[u(\bfn/n)\Q(\bfX_t,\bfn)|\, \bfX_0 = \bfx\right] &= \bbE\left[\bbE[u(\bfN/n)|\, \bfX_t]|\, \bfX_0 = \bfx\right]\\
&\to \bbE[u(\bfX_t)|\bfX_0 = \bfx]
\end{align*}
as $n\to\infty$, where $\bfN \sim \text{Multinomial}(n,\bfX_t)$ and the interchange of limit and integral is justified by \citet{bar:etal:2000}. 
Similarly, 
\[
\sum_{\bfk\in\nabla_{E,n}}\binom{n}{\bfn}m(\bfn) u\left(\frac{\bfk}{n}\right) \to \bbE[u(\bfX_\infty)], \qquad n\to\infty.
\]
These arguments can be shown still to hold if $u$ is replaced by a delta function at $\bfy$ \citep{bar:etal:2000}, and then $\bbE[u(\bfX_t)|\, \bfX_0 = \bfx] = f(\bfx,\bfy;t)$ and $\bbE[u(\bfX_\infty)] = \pi(\bfy)$. Put all this together 
and let $n\to\infty$ to yield \eqref{eq:transition}.
\end{proof}
Equation \eqref{eq:transition} has an intuitive interpretation via Bayes' theorem (\fref{fig:transition}), similar to the one given in \sref{sec:warmup}. The conditional density of $\bfy| \bfx$ is proportional to its prior density $\pi(\bfy)$ times the conditional density of $\bfx| \bfy$. The information that $\bfy$ transfers to the conditional density of $\bfx$ flows through the dual process $\widetilde{\bfL}$, which evolves back from an initial state $\bfy$ to a state $\bfl$ a time $t$ ago (with probability $p_{\bfy\bfl}(t)$). 
Given $\widetilde{\bfL}_t = \bfl$, the density of $\bfx$ is proportional to the likelihood of the type configuration associated with $\bfl$ given $\bfx$ (contributing the multinomial term). The normalisation of this conditional likelihood is the marginal likelihood $\widetilde{m}(\bfl)$ of $\bfl$ when $\bfx$ is integrated over its (prior) stationary distribution.

\begin{figure}[t]
\begin{center}
\includegraphics[width=0.98\textwidth]{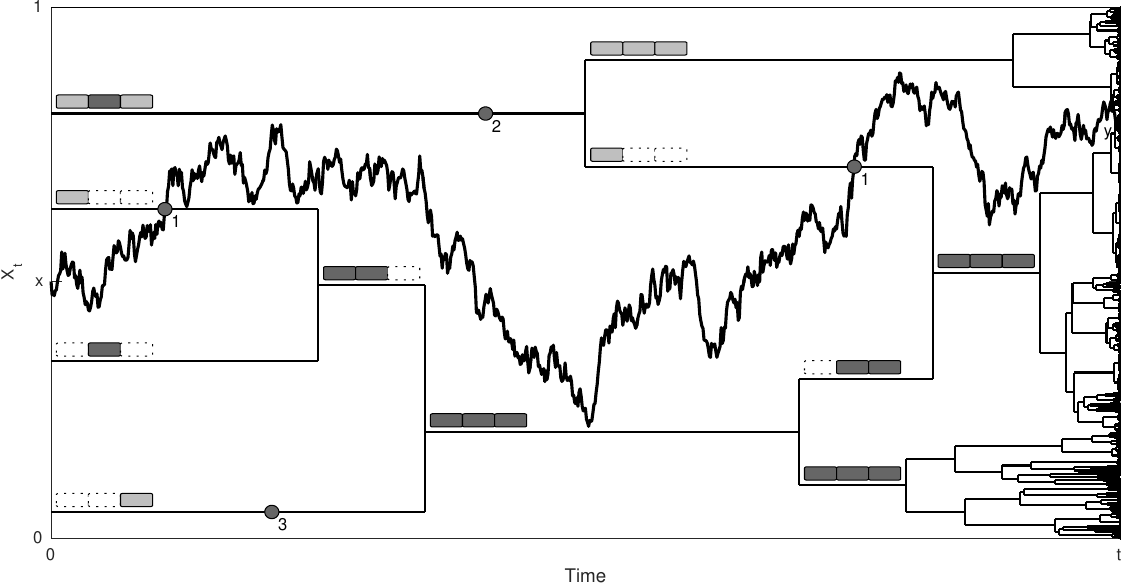}
\end{center}
\caption{\label{fig:transition}Illustration of the transition density in an $L = 3$ locus model, with two alleles at each locus (shown in dark and light grey). The diffusion $\bfX$ evolves from $\bfX_0 = \bfx$ to $\bfX_t = \bfy$ (only the first co-ordinate is plotted). The dual jump process, $\widetilde{\bfL}$, shown here as a typed {\sc arg}, evolves back in time from infinitely many lineages at time $t$ with configuration $\widetilde{\bfL}_0 = \bfy$ to a configuration $\widetilde{\bfL}_t$ of size 4 at time 0 (note the time index now runs backwards). The haplotype associated with each lineage is shown as three shaded segments, and non-ancestral loci are shown with a dotted outline. Mutations in the graph are shown as circles and are labelled by the locus they affect. Denoting the dark and light alleles by 0 and 1 respectively, the four types $(\bfi,A)$ of $\widetilde{\bfL}_t$ are, from top to bottom, $((1,0,1), \{1,2,3\}), ((1), \{1\}), ((0), \{2\}), ((1),\{3\})$.}
\end{figure}

\begin{rem}
For reversible diffusions one can obtain a version of the transition density more flexible than \eqref{eq:transition}, expressed in terms of $p_{\bfx\bfl}(t)$ rather than $p_{\bfy\bfl}(t)$. Despite the interchange of $\bfx$ and $\bfy$, it is still possible to interpret the alternative form for the transition density in terms of a dual process running backwards in time \citep{don:tav:1987, eth:gri:2009}. However, the Wright-Fisher diffusion with recombination is not reversible \citep{han:2002}.
\end{rem}
\begin{rem}
The existence of $p_{\bfy\bfl}(t)$ in a model incorporating selection rather than recombination is proven rigorously by \citet{bar:etal:2000}. It may be possible to adapt their approach here; we leave this for future work.
\end{rem}

\section{A continuous model}
\label{sec:continuous}
Before studying the $L$-locus model further, we illustrate how the above strategy can also be applied to a continuous model of recombination. For this to make sense the mutation model should also be continuous, and an appropriate choice is the infinitely-many-sites model. One way to achieve the appropriate duality result is first to write down the relevant diffusion model and then to pursue the strategy above, for example by recasting it as a Fleming-Viot measure-valued diffusion along the lines of \citet{eth:gri:1987}. Here we take a more direct approach by taking the formal limit in the $L$-locus model as $L\to\infty$. To take this limit painlessly we will reformulate our $L$-locus model somewhat.

First consider a representation for the continuous limit. Here a chromosome is idealised as the interval $[0,1]$, and the model is specified by two probability measures on $[0,1]$, which we assume to admit densities $\eta$ and $\nu$ with respect to Lebesgue measure, respectively modelling the distribution of mutation and recombination events along a chromosome. (The usual infinitely-many-sites model of mutation is recovered by letting $\eta(x) \equiv 1$. This is also a typical choice for $\nu$.) A haplotype in this model can be specified by a set $\bfxi \subseteq [0,1]$ of positions at which it differs from some reference haplotype. If the reference haplotype is chosen to be that of the grand most recent common ancestor of a sample of $n$ haplotypes, then $|\bfxi|$ is finite \citep{gri:mar:1997}. 
The state space for this model is
\[
\Xi_{[0,1],n} := \left\{\bfn = (n_{\bfi}^A)_{\emptyset\neq A\subseteq [0,1],\bfxi\subseteq A}: |\bfxi| < \infty, n_{\bfxi}^A \in \bbN,\sum_{\emptyset\neq A\subseteq [0,1]}\sum_{\bfxi\subseteq A} n_{\bfxi}^A = n\right\},
\]
with each $A$ Borel measurable.

We embed the $L$-locus model 
in this continuous description by the mapping $[L] \mapsto \left\{\frac{1}{L},\frac{2}{L},\dots, 1\right\}$. Then a mutation at locus $l$, or a recombination between locus $l$ and $l+1$, occurs at position $l/L$, and we choose
\[
\begin{array}{cccc}
E_l = \{1,2\}, & \bfP^{(l)} = \left(\begin{matrix} 0 & 1 \\ 1 & 0 \end{matrix}\right), & \ds\gq_l = \gq\int_{\frac{l-1}{L}}^{\frac{l}{L}} \nu(x)\ud x, & \ds\gr_l = \gr\int_{\frac{l-1}{L}}^{\frac{l}{L}} \eta(x) \ud x,
\end{array}
\]
for each $l\in [L]$. In \ref{app:continuous} we show that if we let $L\to\infty$ then this embedding recovers a well-defined limiting process for the dual, with state space $\Xi_{[0,1],n}$, and with a mixture of diffuse and atomic jump kernels. It can be described as follows. Given that the process is currently in state $\bfn \in \Xi_{[0,1],n}$:
\begin{description}
\item[Coalescence.] For each $A,B\subseteq [0,1]$ and $\bfxi\subseteq A\cup B$, the process jumps to $\bfn-\bfe_{\bfxi}^A-\bfe_{\bfxi}^B+\bfe_{\bfxi}^{A\cup B}$ at rate
\[
\frac{1}{2} n_{\bfxi}^{A}(n_{\bfxi}^B - \gd_{AB})\frac{\widetilde{m}(\bfn-\bfe_{\bfxi}^A-\bfe_{\bfxi}^B+\bfe_{\bfxi}^{A\cup B})}{\widetilde{m}(\bfn)}.
\]
\item[Mutation.] For each $A \subseteq [0,1]$ and $\bfxi \subseteq A$, the process jumps at rate
\begin{equation}
\label{eq:mutation}
\frac{\gq}{2} n_{\bfxi}^A \int_{A} \frac{\widetilde{m}(\bfn-\bfe_{\bfxi}^A+\bfe_{\bar{\bfxi}(x)}^{A})}{\widetilde{m}(\bfn)} \eta(x) \ud x,
\end{equation}
where
\begin{equation}
\label{eq:barxi}
\bar{\bfxi}(x) = \begin{cases}
 \bfxi \setminus \left\{x\right\} & \text{if } x \in \bfxi,\\[10pt]
\bfxi \cup \left\{x\right\} & \text{if }x \notin \bfxi.
\end{cases}
\end{equation}
The resulting state is $\bfn-\bfe_{\bfxi}^A+\bfe_{\bfxi \setminus \{x\}}$, where the position $x \in \bfxi$ is chosen by the probability distribution proportional to $\frac{\widetilde{m}(\bfn-\bfe_{\bfxi}^A+\bfe_{\bfxi \setminus \{x\}}^{A})}{\widetilde{m}(\bfn)} \eta(x)\ud x$.
\item[Recombination.] For each $A \subseteq [0,1]$ and $\bfxi \subseteq A$, the process jumps at rate 
\[
\frac{\gr}{2} n_{\bfxi}^A \int_{\inf A}^{\sup A} \frac{\widetilde{m}(\bfn-\bfe_{\bfxi}^A+\bfe_{\bfxi}^{\Alx}+\bfe_{\bfxi}^{\Agx})}{\widetilde{m}(\bfn)} \nu(x) \ud x,
\]
where $\Alx = A \cap [0,x]$ and $\Agx = A \cap (x,1]$. The resulting state is $\bfn-\bfe_{\bfxi}^A+\bfe_{\bfxi}^{\Alx}+\bfe_{\bfxi}^{\Agx}$, with $x \in [\inf A, \sup A]$ chosen by the probability distribution proportional to $\frac{\widetilde{m}(\bfn-\bfe_{\bfxi}^A+\bfe_{\bfxi}^{\Alx}+\bfe_{\bfxi}^{\Agx})}{\widetilde{m}(\bfn)} \nu(x)\ud x$.
\end{description}

In this description, $\widetilde{m}(\cdot)$ is the limit as $L\to\infty$ of \eqref{eq:moments}, in a sense made more precise in \ref{app:continuous}. Since $|\bfxi|$ is finite, the jump distribution due to mutation has finite support. As is shown in \ref{app:continuous}, it is further concentrated on transitions to states of the form $\widehat{\bfn} = \bfn-\bfe_{\bfxi}^A+\bfe_{\bfxi \setminus \{x\}}^{A}$ such that $x \notin \bfmath{\zeta}$ for any $\bfmath{\zeta}$ and $B$ with $\widehat{n}_{\bfmath{\zeta}}^B > 0$ (i.e.\ if a mutation occurs at $x$ then in the resulting configuration no haplotype carries the mutant allele at site $x$---the process obeys the infinitely-many-sites assumption).


\section{The case of no mutation}
\label{sec:nomutation}
As noted in the Introduction, it is possible to make further progress in the absence of mutation. Here we study in further detail the (reduced) $L$-locus model with $\gq = 0$. 
One must take care; the diffusion is no longer ergodic and the stationary distribution is not unique. In fact any distribution placing all its mass at $\gd_\bfj$ for some $\bfj \in E$ is an invariant distribution for $\bfX$; one haplotype $\bfj$ ultimately becomes fixed in the population, and once the diffusion hits this state it stays there. Nevertheless, for each invariant distribution we can find a non-trivial dual process. Here we adapt the results of \sref{sec:reduced}. In order to normalise the duality function of \eqref{eq:F3} with respect to $\bfX_\infty \sim \gd_\bfj$, it is clear that $n_\bfi^A$ can be nonzero only if $\bfi = \bfj|_A$, and then \eqref{eq:F3} simplifies to
\[
\widetilde{F}(\bfx,\bfn) = \prod_{\emptyset\neq A\subseteq [L]} (x_{\bfj}^{A})^{n^A}.
\]
From this one immediately obtains the transition rates of the dual process:
\begin{description}
\item[Coalescence.] For each nonempty $A,B\subseteq [L]$, the process jumps to $\bfn-\bfe_\bfj^A - \bfe_\bfj^B + \bfe_\bfj^{A\cup B}$ at rate
\[
\frac{1}{2}n^A(n^B - \gd_{AB}).
\]
\item[Recombination.] For each nonempty $A\subseteq [L]$ and $l= \min A,\dots, \max A -1$, the process jumps to $\bfn - \bfe_\bfj^A +\bfe_\bfj^{\Al} + \bfe_\bfj^{\Ag}$ at rate
\[
\frac{1}{2}n^A\gr_l.
\]
\end{description}
The state space is $\{\bfn \in \Xi_{E,n} : n_{\bfi}^A = 0 \text{ if } \bfi \neq \bfj|_{A}\}$. This process describes the way that ancestral material is dispersed across the ancestors of a sample. It is the number of lineages in a (reduced, $L$-locus) {\sc arg}. For $L=2$, the dynamics of this process are studied by, for example, \citet{gri:1991} and \citet{sim:chu:1997}. Note that the degree of $\bfn$ is non-increasing, and, assuming that each locus is represented at least once in the initial sample, the process reaches a stationary state with support $\{\text{degree}(\bfn) = L\}$ (each locus has precisely one ancestor), with $\{n = 1\}$ a recurrent set (one individual is simultaneously ancestral at all loci). Starting from a single individual, ancestral material fragments back in time across many different individuals, before almost surely reconvening again within a single ancestor. \citet{ess:etal:2016} call this the \emph{partitioning process} in the context of the Moran model. In the same context, \citet{bob:etal:2010} study its rate of convergence to stationarity and provide a computer program to compute its transient distribution. \citet{wiu:hei:1997} study the process in the context of the continuous model of \sref{sec:continuous}, where they use it to address the question of how many genetic ancestors there are to a contemporary human chromosome.


It is convenient to denote the partitions directly. That is, if $\widetilde{\bfL}$ evolves as a partitioning process (with $\text{degree}(\widetilde{\bfL}_0) = L$), then let $\Theta_t = \{A\subseteq [L]: \widetilde{L}_t^A = 1\}$. Further writing
\[
x_\bfj^{\Theta} = \prod_{A \in \Theta} x_\bfj^A,
\]
for a partition $\Theta$, the duality equation can be written concisely as
\begin{equation}
\label{eq:partitiondual}
\bbE\left[ (X_{\bfj}^{\Phi})_t\mid \bfX_0 = \bfx\right] = \bbE\left[x_{\bfj}^{\Theta_t}\mid \Theta_0 = \Phi\right].
\end{equation}
It relates two particularly important quantities. Expectation on the left-hand side is with respect to $\bfX$ evolving forward in time according to \eqref{eq:WFGeneratorR} (with $\gq = 0$). The left-hand side is therefore a transient moment of the Wright-Fisher diffusion involving combinations of the alleles comprising the haplotype $\bfj$, where the combinations of interest are specified by a partition $\Phi$. Expectation on the right-hand side is with respect to $\bfTheta = (\Theta_t)_{t\geq 0}$ evolving backward in time from $\Phi$. The right-hand side is therefore the {\sc pgf} for the configuration of lineages in a reduced {\sc arg}. 
\citet{man:2013} uses the relationship between these quantities to find, among other things, the probability distribution of $\Theta_t$ for $L = 2$. Via a change of co-ordinate system, \citet{ess:etal:2016} find the distribution of $\Theta_t$ for $L = 3$. 

Letting $t\to\infty$ in \eqref{eq:partitiondual} is also instructive. We find
\begin{equation}
\label{eq:partitionstationary}
\bbE\left[ (X_{\bfj}^{\Phi})_\infty\mid \bfX_0 = \bfx\right] = \bbE\left[x_{\bfj}^{\Theta_\infty}\mid \Theta_0 = \Phi\right].
\end{equation}
The left-hand side of \eqref{eq:partitionstationary} is
\begin{align*}
\bbE\left[ (X_{\bfj}^{\Phi})_\infty\mid \bfX_0 = \bfx\right] &= \bbP[(X^A_\bfj)_\infty = 1, \forall A \in \Phi\mid \bfX_0 = \bfx]\\
 &= \bbP[\bfX_\infty = \bfe_\bfj\mid \bfX_0 = \bfx], 
\end{align*}
the probability that the haplotype $\bfj$ ultimately fixes in the population, starting from initial frequencies $\bfx$. The right-hand side of \eqref{eq:partitionstationary} is the {\sc pgf} of $\Theta_\infty$, the stationary distribution of the partitioning process. Notice that both sides of \eqref{eq:partitionstationary} are independent of $\Phi$. Notice also that, although the left-hand side is conditioned on the initial frequencies $\bfx$ of all haplotypes, it is only terms of the form $x_\bfj^A$ which are needed---the marginal frequency of haplotypes agreeing with $\bfj$ at a subset $A$ of loci. Frequencies of alleles not appearing in $\bfj$ are immaterial (except through their aggregate frequency, which is expressible in terms of $x_\bfj^A$). Thus, for the purpose of computing \eqref{eq:partitionstationary}, at each given locus $l$ one could aggregate all alleles not equal to $j_l$ and treat them as a single type with frequency $1-x_{j_l}^{\{l\}}$.

The above reasoning motivates our interest in $\Theta_\infty$ in providing multilocus fixation probabilities. Let us spell this out further. First note that the fixation probability can be expressed as
\[
\bbP[\bfX_\infty = \bfe_\bfj\mid \bfX_0 = \bfx] = \sum_\Phi \bbF(\Phi)x_\bfj^\Phi,
\]
where $\bbF(\Phi)$ is the probability that there are $|\Phi|$ single individuals whose descendents cause the haplotype $\bfj$ to fix according to the partition $\Phi$; that is, if $\phi_k$ is the $k$th block of $\Phi$ then the $k$th of the $|\Phi|$ individuals is the ancestor to the whole population at the loci in $\phi_k$, and this individual has haplotype in agreement with $\bfj$ at these loci. Writing out both sides of \eqref{eq:partitionstationary},
\begin{equation}
\label{eq:fixsum}
\sum_{\Phi}\bbF(\Phi)x_\bfj^\Phi = \sum_{\Phi} \bbP(\Theta_\infty = \Phi)x_\bfj^\Phi,
\end{equation}
and therefore
\begin{equation}
\label{eq:fix}
\bbF(\Phi) = \bbP(\Theta_\infty = \Phi).
\end{equation}
We emphasise that \eqref{eq:fix} is a nice consequence of duality. In words, the stationary probability that the ancestors of the population partition the loci according to $\Phi$ is equal to the probability that $|\Phi|$ individuals fix according to the partition $\Phi$. This argument could be extended to a continuous model of a gene as in \sref{sec:continuous}, in which case $\Phi$ is a partition of $[0,1]$.

Consider as a simple example the case of $L=2$ loci. There are two possible partitions, $\{\{1,2\}\}$ and $\{\{1\},\{2\}\}$. Numbering these states as 1 and 2, the transition rate matrix 
of $\bfTheta$ is 
\[
\widetilde{\bfQ} = \begin{pmatrix}
-\rho_1/2 & \rho_1/2\\
1 & -1
\end{pmatrix}.
\]
The distribution of $\Theta_\infty$ is the unit solution $\bfpi$ to $\bfpi\widetilde{\bfQ} = \bfzero$, which is easily verified to be
\[
\bfpi = \left(\frac{2}{2+\rho_1},\frac{\rho_1}{2+\rho_1}\right).
\]
The right-hand side of \eqref{eq:partitionstationary} is
\[
\frac{2}{2+\rho_1}x_\bfj + \frac{\rho_1}{2+\rho_1}x_{\bfj}^{\{1\}}x_{\bfj}^{\{2\}},
\]
and by duality this is the probability of fixation of $\bfj$ when initial frequencies are $\bfx$. If the population is initially at linkage equilibrium, so that $x_\bfj = x_\bfj^{\{1\}}x_\bfj^{\{2\}}$, then \eqref{eq:partitiondual} becomes
\[
\bbE[(X_\bfj^\Phi)_t\mid \bfX_0 = \bfx] = x_{\bfj}^{\{1\}}x_{\bfj}^{\{2\}},
\]
because $x_\bfj^{\Theta_t} = x_\bfj^{\{1\}}x_\bfj^{\{2\}}$ for all $\Theta_t$. This agrees with our intuition that fixation probabilities of the two loci are independent when the initial state is one of linkage equilibrium. Of course, a similar statement can be made for more than two loci.

The stationary distribution of $\Theta$ for $L = 3$ loci is given by \citet{wiu:hei:1997}, and its transient dynamics are studied by \citet{ess:etal:2016}, who also found an analogue of \eqref{eq:fix} for a two-locus Moran model.

\subsection{The stationary distribution of the partitioning process}
While the stationary distribution $\bfpi$ of $\Theta_\infty$ is of interest, solving $\bfpi\widetilde{\bfQ} = \bfzero$ may not be straightforward because the size of this linear system grows rapidly with $L$. 
More precisely, the state space for $\Theta_t$ is the set of partitions of $[L]$. The number of such partitions is $B_L$, the $L$th Bell number, which grows at least exponentially with $L$. In this subsection we show how one can compute the stationary distribution of $\Theta_\infty$ by solving a much smaller system, provided one has already computed the corresponding solution for an $(L-1)$-locus system. In this subsection we will use the superscript $(L)$ to denote the dependence on $L$.

The key idea is to consider the collection of indicators $\bfep^{(L)} := (\ge_{ij})_{i,j\in [L]}$ defined by
\[
\ge_{ij} = \begin{cases}
1 & \text{if $i$ and $j$ are in the same block of $\Theta_\infty$,} \\
0 & \text{otherwise}.
\end{cases}
\]
Then $\bfpi^{(L)}$ is expressible as a vector of joint moments of $\bfep^{(L)}$. For example, if $L=2$ then $\bfpi^{(2)} = \bbE(\ge_{12}, 1 - \ge_{12})$. If $L=3$ then
\begin{multline}
\label{eq:3locus}
{\bfpi^{(3)}}' = \begin{pmatrix}
\bbP(\Theta_\infty = \{\{1,2,3\}\})\\
\bbP(\Theta_\infty = \{\{1,2\},\{3\}\})\\
\bbP(\Theta_\infty = \{\{1,3\},\{2\}\})\\
\bbP(\Theta_\infty = \{\{1\},\{2,3\}\})\\
\bbP(\Theta_\infty = \{\{1\},\{2\},\{3\}\})
\end{pmatrix}
= \bbE
\begin{pmatrix}
\ge_{12}\ge_{23}\\
\ge_{12}(1-\ge_{23})\\
\ge_{13}(1-\ge_{12})\\
(1-\ge_{12})\ge_{23}\\
(1-\ge_{12})(1-\ge_{13})(1-\ge_{23})
\end{pmatrix}\\
= \bbE
\begin{pmatrix}
\ge_{12}\ge_{23}\\
\ge_{12} - \ge_{12}\ge_{23}\\
\ge_{13} - \ge_{12}\ge_{23}\\
\ge_{23} - \ge_{12}\ge_{23}\\
1-\ge_{12} - \ge_{13} - \ge_{23} + 2\ge_{12}\ge_{23}
\end{pmatrix}.
\end{multline}
Some of the terms on the right-hand side of \eqref{eq:3locus} are known from the two-locus solution:
\begin{equation}
\label{eq:2locus}
\begin{array}{ccc}
\ds\bbE[\ge_{12}] = \frac{2}{2 + \gr_1}, & \ds\bbE[\ge_{23}] = \frac{2}{2+\gr_2}, & \ds\bbE[\ge_{13}] = \frac{2}{2+\gr_1 + \gr_2}.
\end{array}
\end{equation}
Substituting these results into $\bfpi^{(3)}\widetilde{\bfQ}^{(3)} = \bfzero$, the number of unknowns is reduced from $B_3 = 5$ down to just one, $\bbE[\ge_{12}\ge_{23}]$.

This idea extends to $L$ loci. Suppose we have found $\bfpi^{(L-1)}$; then we know all required joint moments of $\bfep^{(L-1)}$. The sequence $(\bfep^{(L)})_{L=1,2,\dots}$ has an important consistency property: the marginal joint moments of $\bfep^{(L)}$ involving only the indices $1,2,\dots,L-1$ coincide with those of $\bfep^{(L-1)}$. Furthermore, by rescaling the recombination rate across any missing loci, we also know all the necessary joint moments of $\bfep^{(L)}$ involving indices with at most $L-1$ distinct entries in $1,2,\dots,L$. For example, by ``forgetting'' locus 2 we obtain $\bbE[\ge_{13}]$ in \eqref{eq:2locus} by treating loci 1 and 3 as conforming to a two-locus model with recombination parameter $(\gr_1 + \gr_2)/2$. After exploiting this consistency property, the number of remaining unknown terms in $\bfpi^{(L)}\widetilde{\bfQ}^{(L)} = \bfzero$ is, we claim, equal to
\begin{equation}
\label{eq:singleton}
S_L := (-1)^L + \sum_{k=1}^{L}(-1)^{k-1}B_{L-k}.
\end{equation}
To see this, note that each unknown moment is of the form $\bbE[\ge_{i_1j_1}\ge_{i_2j_2}\cdots \ge_{i_dj_d}]$ in which each index $1,2,\dots,L$ appears at least once (otherwise we could appeal to the $(L-1)$-locus solution). Since each index is represented at least once, $\ge_{i_1j_1}\ge_{i_2j_2}\cdots \ge_{i_dj_d}$ defines a partition on $[L]$; that is, $\bbE[\ge_{i_1j_1}\ge_{i_2j_2}\cdots \ge_{i_dj_d}]$ corresponds uniquely to one entry in $\bfpi^{(L)}$ (for example, when $L=3$ we see from \eqref{eq:3locus} that $\bbE[\ge_{12}\ge_{23}]$ is the first entry of $\bfpi^{(3)}$). Moreover, this partition contains no singleton blocks, because any index $i_k$ is paired in a block with some $j_k$. Thus, the number of unknown moments is equal to the number of partitions of $[L]$ containing no singleton blocks, which is given by \eqref{eq:singleton} \citep[A000296 of][and references therein]{oeis:2011}. By substituting known results from the $(L-1)$-locus solution for $\bfep^{(L-1)}$ into $\bfpi^{(L)}\widetilde{\bfQ}^{(L)} = \bfzero$ written in terms of moments of $\bfep^{(L)}$, the system is reduced from $B_L$ to $S_L$ equations, though $S_L$ still exhibits exponential growth in $L$. The first few of these numbers are given in \tref{tab:bell}.

\begin{table}[t]
\caption{\label{tab:bell} The number $B_L$ of partitions of $[L]$, and the number $S_L$ of partitions of $[L]$ containing no singleton blocks.}
\begin{center}
\begin{tabular}{rrr}
\hline
$L$ & $B_L$ & $S_L$\\
\hline 
1 &     1 &     0 \\
2 &     2 &     1 \\
3 &     5 &     1 \\
4 &    15 &     4 \\
5 &    52 &    11 \\
6 &   203 &    41 \\
7 &   877 &   162 \\
8 &  4140 &   715 \\
9 & 21147 &  3425 \\
10 &     115975 & 17722\\\hline
\end{tabular}
\end{center}
\end{table}

The above argument allows for the efficient computation of $\bfpi^{(L)}$ successively for each $L$. The stationary distribution $\bfpi^{(L)}$ is shown in \fref{fig:fragpiL} for $L=1,2,\dots,6$, summarised by the stationary number of blocks $|\Theta_\infty|$ of $\Theta_\infty$. The complete solution for $\bfpi^{(6)}$ is plotted in \fref{fig:fragpi6} for a symmetric recombination model with $\gr_1 = \gr_2 = \dots = \gr_5$. (Interestingly, the mode of $\bfpi^{(6)}$ appears to be either $\{\{1,2,3,4,5,6\}\}$ or $\{\{1\},\{2\},\{3\},\{4\},\{5\},\{6\}\}$ for any value of $\gr_l$.) We note that these observations are consistent with similar ones made by \citet[Section 4.1]{bob:etal:2010}, who investigated $\Theta_\infty$ for a discrete-time Moran model by numerically iterating the partitioning process over generations until convergence to a chosen precision.

Duality tells us that fixation probabilities can be obtained as certain linear combinations of the curves in \fref{fig:fragpi6}. For example, suppose the population is fixed for a wild-type allele at each of the six loci. At each locus a mutant appears on the wild-type background and its haplotype drifts to frequency $1/6$ (this might be thought of as a haplotype frequency configuration of maximal Hill-Robertson-type interference, though here everything is neutral). What is the probability that all six mutant alleles ultimately fix? Letting $\bfj$ denote the haplotype comprised of all six mutant alleles, from \eqref{eq:fixsum} the only partition $\Phi$ for which $x_{\bfj}^\Phi$ is nonzero is $\Phi = \{\{1\},\{2\},\{3\},\{4\},\{5\},\{6\}\}$. Thus, \eqref{eq:fixsum} tells us that the fixation probability for $\bfj$ is given by the stationary probability of $\Phi$ (the dashed line in \fref{fig:fragpi6}) times $x_\bfj^\Phi = \left(\frac{1}{6}\right)^6$. So in this example, the dashed curve in \fref{fig:fragpi6} also provides the fixation probability of $\bfj$ relative to the completely unlinked case, $\gr_l = \infty$.

\section{Discussion}
\label{sec:discussion}
This paper makes three main contributions. First, we constructed the first duality relationships for population genetics models involving all of genetic drift, mutation, and recombination. They make precise the link between two individually well studied objects; namely, the Wright-Fisher diffusion with recombination and the {\sc arg}. This is done first for a discrete model of recombination and mutation and later on for a continuous limit model. Second, we emphasise the methods underlying our approach: it is particularly algebraically efficient to express the duality of two processes through their infinitesimal generators and to apply those generators to appropriate \emph{generating functions}. Furthermore, this method is fairly automatic and avoids the pitfalls of the probabilistic arguments that are often invoked to address these types of questions. The price for this, one might argue, is that a biological interpretation of the results may be obscured. In this paper we have attempted to spell out how such biological interpretations can be recovered, by distilling mathematical expressions where possible to simple interpretable statements about conditional evolution. Third, we have highlighted the usefulness of our results via two applications: we obtained an expression for the transition function of the diffusion, and we showed how the partitioning process that arises when mutation is ignored can be related to predictions for haplotype fixation probabilities.

\begin{figure}[t]
\begin{tabular}{cc}
(a) $\rho_l = 5$ & (b) $\rho = 5$\\
\includegraphics[width=0.475\textwidth]{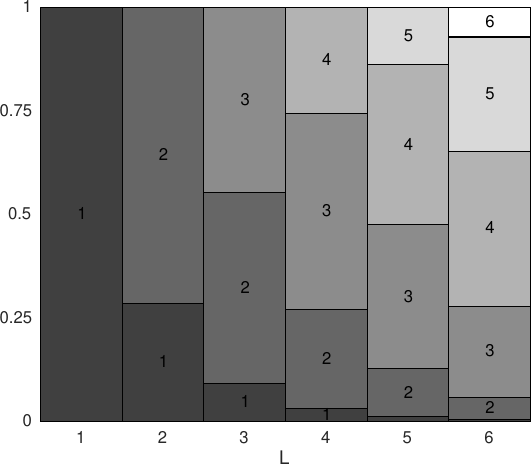} &
\includegraphics[width=0.475\textwidth]{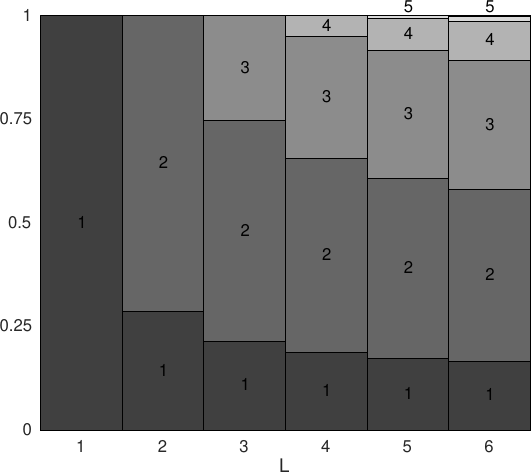} 
\end{tabular}
\caption{\label{fig:fragpiL}Stationary distribution of the number of fragments, $|\Theta_\infty|$, in an $L$-locus model. (a): Fixed per-locus recombination rate, $\rho_l = 5$. (b): Fixed total recombination rate, $\rho = \sum_{l=1}^{L-1}\rho_l = 5$.}
\end{figure}

\begin{figure}[p]
  \centering
  \subfigure[]{
    \includegraphics[width=\textwidth]{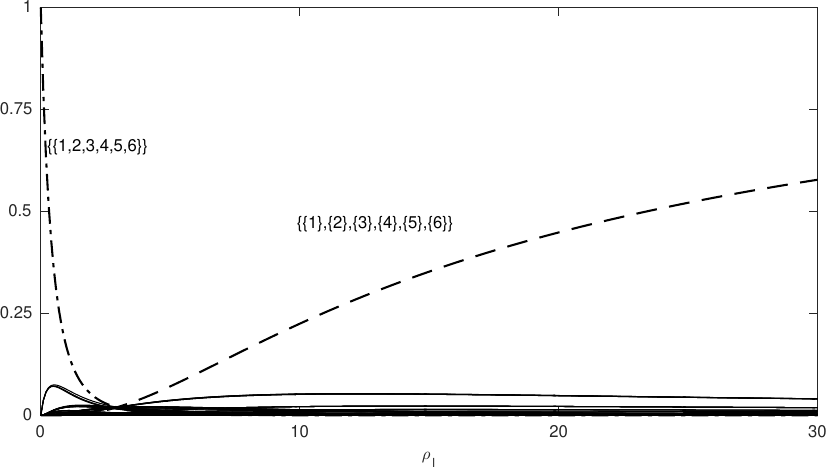}
  }
  \subfigure[]{
    \includegraphics[width=\textwidth]{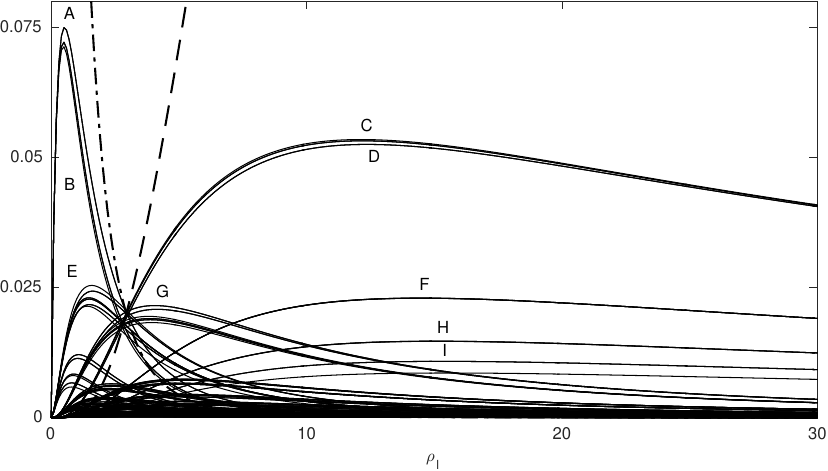}
  }
\end{figure}

\begin{figure}[p]
  \centering
     \subfigure[]{
   \begin{tabular}{cl}
\hline
A & $\{\{1,2,3,4,5\},\{6\}\}$, $\{\{1\},\{2,3,4,5,6\}\}$.\\
B & $\{\{1,2,3,4\},\{5,6\}\}$, $\{\{1,2\},\{3,4,5,6\}\}$, $\{\{1,2,3\},\{4,5,6\}\}$.\\
C & $\{\{1\},\{2\},\{3,4\},\{5\},\{6\}\}$, $\{\{1\},\{2,3\},\{4\},\{5\},\{6\}\}$, \\ &$\{\{1\},\{2\},\{3\},\{4,5\},\{6\}\}$.\\
D & $\{\{1,2\},\{3\},\{4\},\{5\},\{6\}\}$, $\{\{1\},\{2\},\{3\},\{4\},\{5,6\}\}$.\\
E & $\{\{1\},\{2,3,4,5\},\{6\}\}$.\\
F & $\{\{1\},\{2\},\{3\},\{4,6\},\{5\}\}$, $\{\{1,3\},\{2\},\{4\},\{5\},\{6\}\}$, \\ &$\{\{1\},\{2,4\},\{3\},\{5\},\{6\}\}$, $\{\{1\},\{2\},\{3,5\},\{4\},\{6\}\}$.\\
G & $\{\{1\},\{2,3,4\},\{5\},\{6\}\}$, $\{\{1\},\{2\},\{3,4,5\},\{6\}\}$.\\
H & $\{\{1,4\},\{2\},\{3\},\{5\},\{6\}\}$, $\{\{1\},\{2\},\{3,6\},\{4\},\{5\}\}$,\\ & $\{\{1\},\{2,5\},\{3\},\{4\},\{6\}\}$.\\
I & $\{\{1,5\},\{2\},\{3\},\{4\},\{6\}\}$, $\{\{1\},\{2,6\},\{3\},\{4\},\{5\}\}$.\\\hline
\end{tabular}
  }
    \subfigure[]{
    \includegraphics[width=\textwidth]{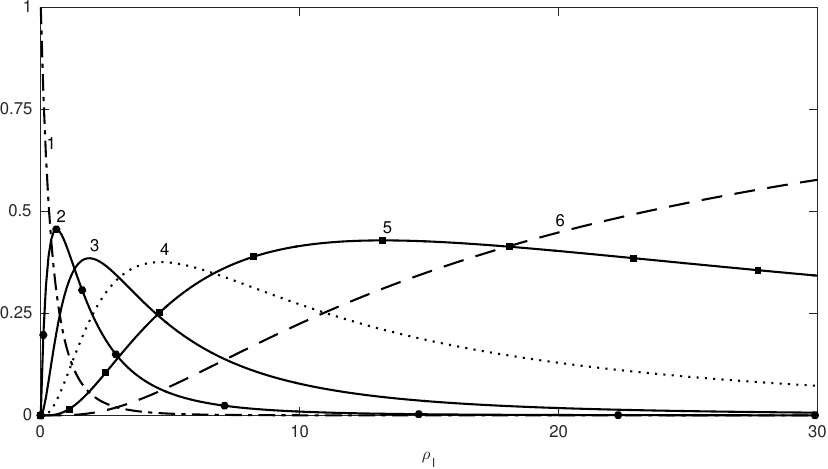}
  }
      \caption{\label{fig:fragpi6}(a): Stationary fragment distribution, $\bfpi^{(L)}$, of $\Theta_\infty$ for an $L=6$ locus model with recombination parameter $\rho_l$ at each breakpoint. (b): A detailed region of (a), with (c): a selection of partitions annotated. (d): The stationary distribution of the number of fragments, $|\Theta_\infty|$.}
\end{figure}

\section*{Acknowledgements}
This work was supported in part by an Engineering \& Physical Sciences Research Council grant to P.A.J.\ (EP/L018497/1). Part of this work was carried out while P.A.J.\ was at the University of California, Berkeley, supported in part by NIH Grant R01-GM094402, and while R.C.G.\ was visiting the D\'epartement de Math\'ematiques et de Statistique at the Universit\'e de Montr\'eal, supported by the Clay Mathematics Institute. He would like to thank his hosts for their hospitality.

\appendix
\section{Useful identitites}
\label{app:identities}
For the function $\widetilde{\Q}(\bfx, \bfn)$ defined by \eqref{eq:Q2} and for $l=\min A,\dots,\max A - 1$, note that
\begin{multline}
x_{\bfi}^{A\cup B}\widetilde{\Q}(\bfx,\bfn - \bfe_{\bfi}^A - \bfe_\bfi^B) = \frac{\binom{n-2}{\bfn - \bfe_\bfi^A - \bfe_\bfi^B}}{\binom{n-1}{\bfn-\bfe_i^A - \bfe_\bfi^B + \bfe_\bfi^{A\cup B}}}\widetilde{\Q}(\bfx,\bfn-\bfe_\bfi^A - \bfe_\bfi^B + \bfe_\bfi^{A\cup B})\\
= \frac{n_{\bfi}^{A\cup B} + 1 -\gd_{A,A\cup B} - \gd_{B,A\cup B}}{n-1}\widetilde{\Q}(\bfx,\bfn-\bfe_\bfi^A - \bfe_\bfi^B + \bfe_\bfi^{A\cup B}), \label{eq:ex1}
\end{multline}
\begin{multline}
x_{\m{\bfi}{l}{j}}\widetilde{\Q}(\bfx,\bfn - \bfe_\bfi^A) = \frac{\binom{n-1}{\bfn-\bfe_\bfi^A}}{\binom{n}{\bfn - \bfe_{\bfi}^A + \bfe_{\m{\bfi}{l}{j}}^A}}\widetilde{\Q}(\bfx,\bfn-\bfe_{\bfi}^A +\bfe_{\m{\bfi}{l}{j}}^A)\\
= \frac{n_{\m{\bfi}{l}{j}}^A+1-\gd_{i_lj}}{n}\widetilde{\Q}(\bfx,\bfn-\bfe_{\bfi}^A +\bfe_{\m{\bfi}{l}{j}}^A),\label{eq:ex2}
\end{multline}
\begin{multline}
x_{\bfi}^{\Al}x_{\bfi}^{\Ag}\widetilde{\Q}(\bfx,\bfn - \bfe_\bfi^A) = \frac{\binom{n-1}{\bfn-\bfe_\bfi^A}}{\binom{n+1}{\bfn - \bfe_\bfi^A + \bfe_\bfi^{\Al} + \bfe_\bfi^{\Ag}}} \widetilde{\Q}(\bfx,\bfn-\bfe_\bfi^A+\bfe_\bfi^{\Al}+\bfe_\bfi^{\Ag})\\
= \frac{(n_\bfi^{\Al}+1)(n_\bfi^{\Ag}+1)}{n(n+1)}\widetilde{\Q}(\bfx,\bfn-\bfe_\bfi^A+\bfe_\bfi^{\Al}+\bfe_\bfi^{\Ag}).\label{eq:ex3}
\end{multline}

\section{The continuous limit}
\label{app:continuous}
In this appendix we show how to recover the continuous dual process described in \sref{sec:continuous} when the $L$-locus model is embedded in it; $E_l$, $\bfP^{(l)}$, $\theta_l$, and $\gr_l$ are defined as in that section, and we let $L\to\infty$. To emphasise the dependence on $L$, in this appendix we will write $\bfn^{(L)}$, $\Xi_{E,n}^{(L)}$, and $\sL^{(L)}$ for $\bfn$, $\Xi_{E,n}$, and $\sL$. In order to identify the limiting behaviour of the process $\widetilde{\bfL}$ of \thmref{thm:reducedFiniteSites}, we proceed by fixing $\bfn \in \Xi_{[0,1],n}$, constructing a sequence $\bfn^{(L)} \in \Xi_{E,n}^{(L)}$ converging to $\bfn$ (in a manner to be defined precisely below), and then seeking the limit of $\sL^{(L)} \widetilde{F}(\bfx,\bfn^{(L)})$ as $L \to \infty$.

To construct a sequence $(\bfn^{(L)})_{L\in\bbN}$ converging to some $\bfn \in \Xi_{[0,1],n}$, we define $\bfn^{(L)}$ as: 
\begin{equation}
\label{eq:sampleconvergence}
n_{\bfxi^{(L)}}^{A^{(L)}} = \sum_{\substack{A\subseteq [0,1]:\\ A^{(L)} = LA \cap [L]}} \,\sum_{\substack{\bfxi \subseteq A:\, |\bfxi| < \infty,\\ \xi_i^{(L)}L = \lceil\xi_iL\rceil,\, i=0,1,\dots}} n_{\bfxi}^A.
\end{equation}
Equation \eqref{eq:sampleconvergence} defines an obvious `coarsening' for representing a sample from the continuous model in its $L$-locus counterpart: the position of each mutant site is rounded up to the nearest multiple of $\frac{1}{L}$, and the segment $A$ over which a haplotype is ancestral is represented by the collection $\{l \in [L]: \frac{l}{L} \in A\} =: A^{(L)}$. Given a sample $\bfn$, for sufficiently large $L$ we have 
\begin{equation}
\label{eq:virtue}
n^{A^{(L)}}_{\bfxi^{(L)}} = n^{A}_{\bfxi}, \qquad \text{for each }\bfxi \subseteq A \subseteq [0,1],
\end{equation}
and we write $\bfn^{(L)} \to \bfn$ as $L\to\infty$. Similarly, we can fix the role of $\bfx$ by choosing $x_{\bfxi^{(L)}}^{A^{(L)}} = x_{\bfxi}^A$ for each $\bfxi$ and $A$ with $n_{\bfxi}^A > 0$.

In this formulation, for sufficiently large $L$ equation \eqref{eq:generatorQ3} becomes: 
{\allowdisplaybreaks
\begin{align}
	\lefteqn{\sL^{(L)} \widetilde{\Q}(\bfx, \bfn^{(L)}) =}\hspace{40pt}&\notag\\
	{}& \frac{1}{2}\sum_{\emptyset\neq A\subseteq [0,1]}\Bigg[\sum_{\emptyset\neq B\subseteq [0,1]}\sum_{\bfxi\subseteq(A\cup B)} n(n_{\bfxi}^{A\cup B} + 1 - \gd_{A,A\cup B} - \gd_{B,A\cup B}) \notag\\
	{}& \specialcell{\hfill \times
	\widetilde{\Q}(\bfx,\bfn-\bfe_{\bfxi}^A-\bfe_{\bfxi}^B+\bfe_{\bfxi}^{A\cup B})} \nonumber\\
	{}& {}+ \gq\sum_{\bfxi \subseteq A} \int_{\bigcup_{l\in A^{(L)}}[\frac{l-1}{L},\frac{l}{L}]} (n_{\bar{\bfxi}(x)}^A + 1)\widetilde{\Q}(\bfx,\bfn-\bfe_{\bfxi}^A+\bfe_{\bar{\bfxi}(x)}^{A}) \eta(x) \ud x \nonumber\\
	 {}& {}+ \gr\sum_{\bfxi\subseteq A}  \int_{\frac{1}{L}(\min A^{(L)} - 1)}^{\frac{1}{L}\max A^{(L)}} \frac{(n_{\bfxi}^{\Alx} + 1)(n_{\bfxi}^{\Agx} + 1)}{n+1}\notag\\
	 {}& \specialcell{\hfill \times\widetilde{\Q}(\bfx,\bfn-\bfe_{\bfxi}^A+\bfe_{\bfxi}^{\Alx}+\bfe_{\bfxi}^{\Agx})\nu(x)\ud x\Bigg]} \nonumber\\
	{}& - \frac{1}{2}\left[n(n-1) + \sum_{A\subseteq [0,1]} n^{A} \left(\gq\int_{\bigcup_{l\in A^{(L)}}[\frac{l-1}{L},\frac{l}{L}]} \eta(x)\ud x  \right.\right.\notag\\
	{}& \specialcell{\hfill \left.\left. {}+ \gr \int_{\frac{1}{L}(\min A^{(L)} - 1)}^{\frac{1}{L}\max A^{(L)}}  \nu(x) \ud x \right)\right]\widetilde{\Q}(\bfx,\bfn),}	\label{eq:generatorQ4}
\end{align}
}
where $\Alx = A \cap [0,x]$, $\Agx = A \cap (x,1]$, and $\bar{\bfxi}(x)$ is given by \eqref{eq:barxi}. [Superscripts illustrating the dependence of $\bfn^{(L)}$ on $L$ 
can be dropped, by virtue of \eqref{eq:virtue}.] We can now take the limit as $L\to\infty$ in \eqref{eq:generatorQ4}; simply replace the range of integration for the mutation terms with $A$, and replace the range of integration for the recombination terms with $[\inf A, \sup A]$. 
In a similar manner, one can reformulate $\sL^{(L)}\widetilde{F}(\bfx,\bfn^{(L)})$ and let $L \to \infty$ to find
{\allowdisplaybreaks
\begin{align}
	\lefteqn{\sL \widetilde{F}(\bfx, \bfn) =}\notag\\
	{}& \frac{1}{2}\sum_{\emptyset\neq A\subseteq [0,1]}\Bigg[\sum_{\emptyset\neq B\subseteq [0,1]}\sum_{\bfxi\subseteq(A\cup B)} n_{\bfxi}^{A}(n_{\bfxi}^B - \gd_{AB})\frac{\widetilde{m}(\bfn-\bfe_{\bfxi}^A-\bfe_{\bfxi}^B+\bfe_{\bfxi}^{A\cup B})}{\widetilde{m}(\bfn)}\notag\\
	{}& \specialcell{\hfill \times
	\widetilde{F}(\bfx,\bfn-\bfe_{\bfxi}^A-\bfe_{\bfxi}^B+\bfe_{\bfxi}^{A\cup B})}\nonumber\\
	{}& {}+ \gq\sum_{\bfxi \subseteq A} n_{\bfxi}^A \int_{A} \frac{\widetilde{m}(\bfn-\bfe_{\bfxi}^A+\bfe_{\bar{\bfxi}(x)}^{A})}{\widetilde{m}(\bfn)} \widetilde{F}(\bfx,\bfn-\bfe_{\bfxi}^A+\bfe_{\bar{\bfxi}(x)}^{A}) \eta(x) \ud x \nonumber\\
	 {}& {}+ \gr\sum_{\bfxi\subseteq A} n_{\bfxi}^A \int_{\inf A}^{\sup A} \frac{\widetilde{m}(\bfn-\bfe_{\bfxi}^A+\bfe_{\bfxi}^{\Alx}+\bfe_{\bfxi}^{\Agx})}{\widetilde{m}(\bfn)} \notag\\
	  {}& \specialcell{\hfill \times\widetilde{F}(\bfx,\bfn-\bfe_{\bfxi}^A+\bfe_{\bfxi}^{\Alx}+\bfe_{\bfxi}^{\Agx})\nu(x) \ud x\Bigg]} \nonumber\\
	{}& - \frac{1}{2}\left[n(n-1)+ \sum_{A\subseteq [0,1]} n^{A} \left(\gq\int_{A} \eta(x)\ud x  + \gr \int_{\inf A}^{\sup A}  \nu(x) \ud x \right)\right]\widetilde{F}(\bfx,\bfn), \label{eq:generatorQ5}
\end{align}
}
where $\widetilde{m}(\widehat{\bfn})/\widetilde{m}(\bfn)$ is defined as the weak limit satisfying
\[
\int_C\frac{\widetilde{m}(\widehat{\bfn})}{\widetilde{m}(\bfn)}\widetilde{F}(\bfx,\widehat{\bfn}) \gl(x) \ud x = \lim_{L\to\infty} \int_C \frac{\widetilde{m}(\widehat{\bfn}^{(L)})}{\widetilde{m}(\bfn^{(L)})}\widetilde{F}(\bfx,\widehat{\bfn}^{(L)})\lambda(x) \ud x,
\]
for $C \subseteq [0,1]$ and $\gl$ a probability density on $[0,1]$. (We refrain from passing the limit through the integral, since in some instances it is necessary to interpret the limit in a Dirac sense; see below.) The interpretation of \eqref{eq:generatorQ5} as the generator of a pure jump Markov process is clear, and the terms corresponding to coalescence and recombination events agree with the description given in \sref{sec:continuous}. The mutation term, however, reads as:
\begin{description}
\item[Mutation.] For each $A \subseteq [0,1]$ and $\bfxi \subseteq A$, the process jumps at rate
\begin{equation*}
\frac{\gq}{2} n_{\bfxi}^A \int_{A} \frac{\widetilde{m}(\bfn-\bfe_{\bfxi}^A+\bfe_{\bar{\bfxi}(x)}^{A})}{\widetilde{m}(\bfn)} \eta(x) \ud x.
\end{equation*}
The resulting state is $\bfn-\bfe_{\bfxi}^A+\bfe_{\bar{\bfxi}(x)}^{A}$, with the position $x \in A$ chosen by the probability distribution proportional to $\frac{\widetilde{m}(\bfn-\bfe_{\bfxi}^A+\bfe_{\bar{\bfxi}(x)}^{A})}{\widetilde{m}(\bfn)} \eta(x)\ud x$.
\end{description}
It remains to reconcile this with the description for mutation given in \sref{sec:continuous}, which follows if we can show that the infinitely-many-sites assumption holds in the limit. More precisely, we should see transitions only to states of the form $\widehat{\bfn} = \bfn-\bfe_{\bfxi}^A+\bfe_{\bfxi \setminus \{x\}}^{A}$ such that $x \in \bfxi$, and such that $x \notin \bfmath{\zeta}$ for any $\bfmath{\zeta}$ and $B$ with $\widehat{n}_{\bfmath{\zeta}}^B > 0$. This holds by the following lemma, from which we deduce that if $\widehat{\bfn}$ is \emph{not} of this form then $\widetilde{m}(\widehat{\bfn}^{(L)})/\widetilde{m}(\bfn^{(L)}) \to 0$ as $L\to\infty$.

\begin{lem}
Let
\[
s(\bfn^{(L)}) = \left|\bigcup_{\bfxi^{(L)} \subseteq A^{(L)}: n_{\bfxi^{(L)}}^{A^{(L)}} > 0} \bfxi^{(L)}\right|
\]
denote the total number of segregating sites in a sample $\bfn^{(L)} \in \Xi_{E,n}^{(L)}$. If $s(\bfn^{(L)}) = O(1)$ then $\widetilde{m}(\bfn^{(L)}) = O(L^{-s(\bfn^{(L)})})$ as $L \to \infty$.
\end{lem}
\begin{proof}
$\widetilde{m}(\bfn^{(L)})$ satisfies the finite system \eqref{eq:SamplingDistributionFiniteSites}, whose solution is unique. (The boundary condition is adjusted to account for our definition of $\bfxi$ with respect to a reference haplotype: $\widetilde{m}(\bfe_{\bfxi^{(L)}}) = \gd_{\bfxi^{(L)}\emptyset}$.) It is straightforward to check that $\widetilde{m}(\bfn^{(L)}) = O(L^{-s(\bfn^{(L)})})$ satisfies this system: The left-hand side, and the first and third terms on the right are all clearly $O(L^{-s(\bfn^{(L)})})$. The second term on the right, corresponding to mutation events, has three contributions: First, there are $O(1)$ summands for which (in the notation of this section) $\widehat{\bfn} = \bfn-\bfe_{\bfxi^{(L)}}^{A^{(L)}} + \bfe_{\bar{\bfxi}^{(L)}(\frac{l}{L})}^{A^{(L)}}$ has one fewer segregating site; these terms contribute $\gq_l \times \widetilde{m}(\widehat{\bfn}) = O(L^{-1} \times L^{-(s(\bfn^{(L)})-1)}) = O(L^{-s(\bfn^{(L)})})$. Second, there are $O(1)$ summands for which $\widehat{\bfn}$ has the same number of segregating sites (parallel mutations); these terms contribute $\gq_l \times \widetilde{m}(\widehat{\bfn}) = O(L^{-1} \times L^{-s(\bfn^{(L)})}) = O(L^{-(s(\bfn^{(L)})+1)})$ and vanish in the limit. Third, there are $O(L)$ summands for which $\widehat{\bfn}$ has one extra segregating site (back mutations); these terms each contribute $O(L^{-1} \times L^{-s(\bfn^{(L)})+1})$ and also vanish in the limit.
\end{proof}

Thus, back mutations are not seen in the limit because the integrand in \eqref{eq:mutation} vanishes, while parallel mutation are not seen because the integrand is $O(1)$ but the range of integration for such events has Lebesgue measure zero. The integral \eqref{eq:mutation} is recognised retrospectively as a sum over at most $|\bfxi|$ atoms. 


\begin{thebibliography}{43}
\expandafter\ifx\csname natexlab\endcsname\relax\def\natexlab#1{#1}\fi
\expandafter\ifx\csname url\endcsname\relax
  \def\url#1{\texttt{#1}}\fi
\expandafter\ifx\csname urlprefix\endcsname\relax\def\urlprefix{URL }\fi

\bibitem[{Barbour et~al.(2000)Barbour, Ethier, and Griffiths}]{bar:etal:2000}
Barbour, A.~D., Ethier, S.~N., Griffiths, R.~C., 2000. A transition function
  expansion for a diffusion model with selection. Annals of Applied Probability
  10~(1), 123--162.

\bibitem[{Bobrowski et~al.(2010)Bobrowski, Wojdy\l{}a, and
  Kimmel}]{bob:etal:2010}
Bobrowski, A., Wojdy\l{}a, T., Kimmel, M., 2010. Asymptotic behavior of a
  {Moran} model with mutations, drift and recombination among multiple loci.
  Journal of Mathematical Biology 61, 455--473.

\bibitem[{Donnelly and Kurtz(1999)}]{don:kur:1999:AAP}
Donnelly, P., Kurtz, T.~G., 1999. Genealogical processes for {Fleming-Viot}
  models with selection and recombination. Annals of Applied Probability 9~(4),
  1091--1148.

\bibitem[{Donnelly and Tavar\'e(1987)}]{don:tav:1987}
Donnelly, P., Tavar\'e, S., 1987. The population genealogy of the
  infinitely-many neutral alleles model. Journal of Mathematical Biology 25,
  381--391.

\bibitem[{Esser et~al.(2016)Esser, Probst, and Baake}]{ess:etal:2016}
Esser, M., Probst, S., Baake, E., 2016. Partitioning, duality, and linkage
  disequilibria in the {Moran} model with recombination. Journal of
  Mathematical Biology 73~(1), 161--197.

\bibitem[{Etheridge and Griffiths(2009)}]{eth:gri:2009}
Etheridge, A.~M., Griffiths, R.~C., 2009. A coalescent dual process in a
  {Moran} model with genic selection. Theoretical Population Biology 75,
  320--330.

\bibitem[{Etheridge et~al.(2010)Etheridge, Griffiths, and
  Taylor}]{eth:etal:2010}
Etheridge, A.~M., Griffiths, R.~C., Taylor, J.~E., 2010. A coalescent dual
  process in a {Moran} model with genic selection, and the lambda coalescent
  limit. Theoretical Population Biology 78, 77--92.

\bibitem[{Ethier and Griffiths(1987)}]{eth:gri:1987}
Ethier, S.~N., Griffiths, R.~C., 1987. The infinitely-many-sites model as a
  measure-valued diffusion. The Annals of Probability 15~(2), 515--545.

\bibitem[{Ethier and Griffiths(1990{\natexlab{a}})}]{eth:gri:1990:AAP}
Ethier, S.~N., Griffiths, R.~C., 1990{\natexlab{a}}. The neutral two-locus
  model as a measure-valued diffusion. Advances in Applied Probability 22~(4),
  773--786.

\bibitem[{Ethier and Griffiths(1990{\natexlab{b}})}]{eth:gri:1990:JMB}
Ethier, S.~N., Griffiths, R.~C., 1990{\natexlab{b}}. On the two-locus sampling
  distribution. Journal of Mathematical Biology 29, 131--159.

\bibitem[{Ethier and Griffiths(1993)}]{eth:gri:1993}
Ethier, S.~N., Griffiths, R.~C., 1993. The transition function of a
  {Fleming-Viot} process. Annals of Probability 21~(3), 1571--1590.

\bibitem[{Ethier and Kurtz(1993)}]{eth:kur:1993}
Ethier, S.~N., Kurtz, T.~G., 1993. {Fleming-Viot} processes in population
  genetics. SIAM Journal of Control and Optimization 31~(2), 345--386.

\bibitem[{Fearnhead(2002)}]{fea:2002}
Fearnhead, P., 2002. The common ancestor at a nonneutral locus. Journal of
  Applied Probability 39, 38--54.

\bibitem[{Fearnhead(2003)}]{fea:2003:JAP}
Fearnhead, P., 2003. Haplotypes: the joint distribution of alleles at linked
  loci. Journal of Applied Probability 40, 505--512.

\bibitem[{Fearnhead and Donnelly(2001)}]{fea:don:2001}
Fearnhead, P., Donnelly, P., 2001. Estimating recombination rates from
  population genetic data. Genetics 159, 1299--1318.

\bibitem[{Golding(1984)}]{gol:1984}
Golding, G.~B., 1984. The sampling distribution of linkage disequilibrium.
  Genetics 108, 257--274.

\bibitem[{Griffiths(1979)}]{gri:1979:AAP11:310}
Griffiths, R.~C., 1979. A transition density expansion for a multi-allele
  diffusion model. Advances in Applied Probability 11~(2), 310--325.

\bibitem[{Griffiths(1980)}]{gri:1980}
Griffiths, R.~C., 1980. Lines of descent in the diffusion approximation of
  neutral {Wright-Fisher} models. Theoretical Population Biology 17, 37--50.

\bibitem[{Griffiths(1981)}]{gri:1981}
Griffiths, R.~C., 1981. Neutral two-locus multiple allele models with
  recombination. Theoretical Population Biology 19, 169--186.

\bibitem[{Griffiths(1991)}]{gri:1991}
Griffiths, R.~C., 1991. The two-locus ancestral graph. In: Basawa, I.~V.,
  Taylor, R.~L. (Eds.), Selected proceedings of the Sheffield symposium on
  applied probability: 18. IMS Lecture Notes---Monograph series. Vol.~18. pp.
  100--117.

\bibitem[{Griffiths et~al.(2008)Griffiths, Jenkins, and Song}]{gri:etal:2008}
Griffiths, R.~C., Jenkins, P.~A., Song, Y.~S., 2008. Importance sampling and
  the two-locus model with subdivided population structure. Advances in Applied
  Probability 40~(2), 473--500.

\bibitem[{Griffiths and Marjoram(1996)}]{gri:mar:1996}
Griffiths, R.~C., Marjoram, P., 1996. Ancestral inference from samples of {DNA}
  sequences with recombination. Journal of Computational Biology 3~(4),
  479--502.

\bibitem[{Griffiths and Marjoram(1997)}]{gri:mar:1997}
Griffiths, R.~C., Marjoram, P., 1997. An ancestral recombination graph. In:
  Donnelly, P., Tavar\'e, S. (Eds.), Progress in population genetics and human
  evolution. Vol.~87. Springer-Verlag Berlin, pp. 257--270.

\bibitem[{Handa(2002)}]{han:2002}
Handa, K., 2002. Quasi-invariance and reversibility in the {Fleming-Viot}
  process. Probability Theory and Related Fields 122, 545--566.

\bibitem[{Hudson(1983)}]{hud:1983}
Hudson, R.~R., 1983. Properties of a neutral allele model with intragenic
  recombination. Theoretical Population Biology 23, 183--201.

\bibitem[{Jansen and Kurt(2014)}]{jan:kur:2014}
Jansen, S., Kurt, N., 2014. On the notion(s) of duality for {Markov} processes.
  Probability Surveys 11, 59--120.

\bibitem[{Jenkins and Griffiths(2011)}]{jen:gri:2011}
Jenkins, P.~A., Griffiths, R.~C., 2011. Inference from samples of {DNA}
  sequences using a two-locus model. Journal of Computational Biology 18~(1),
  109--127.

\bibitem[{Jenkins and Song(2009)}]{jen:son:2009:G}
Jenkins, P.~A., Song, Y.~S., 2009. Closed-form two-locus sampling
  distributions: accuracy and universality. Genetics 183, 1087--1103.

\bibitem[{Kamm et~al.(2016)Kamm, Spence, Chan, and Song}]{kam:etal:2016}
Kamm, J.~A., Spence, J.~P., Chan, J., Song, Y.~S., 2016. Two-locus likelihoods
  under variable population size and fine-scale recombination rate estimation.
  Genetics 203~(3), 1381--1399.

\bibitem[{Kingman(1982)}]{kin:1982:SPA}
Kingman, J. F.~C., 1982. The coalescent. Stochastic Processes and their
  Applications 13~(3), 235--248.

\bibitem[{Krone and Neuhauser(1997)}]{kro:neu:1997}
Krone, S.~M., Neuhauser, C., 1997. Ancestral processes with selection.
  Theoretical Population Biology 51~(3), 210--237.

\bibitem[{Larribe and Lessard(2008)}]{lar:les:2008}
Larribe, F., Lessard, S., 2008. A composite-conditional-likelihood approach for
  gene mapping based on linkage disequilibrium in windows of marker loci.
  Statistical Applications in Genetics and Molecular Biology 7~(1), Article 27.

\bibitem[{Larribe et~al.(2002)Larribe, Lessard, and Schork}]{lar:etal:2002}
Larribe, F., Lessard, S., Schork, N.~J., 2002. Gene mapping via the ancestral
  recombination graph. Theoretical Population Biology 62, 215--229.

\bibitem[{Lohse et~al.(2016)Lohse, Chmelik, Martin, and Barton}]{loh:etal:2016}
Lohse, K., Chmelik, M., Martin, S.~H., Barton, N.~H., 2016. Efficient
  strategies for calculating blockwise likelihoods under the coalescent.
  Genetics 202~(2), 775--786.

\bibitem[{Lohse et~al.(2011)Lohse, Harrison, and Barton}]{loh:etal:2011}
Lohse, K., Harrison, R.~J., Barton, N.~H., 2011. A general method for
  calculating likelihoods under the coalescent process. Genetics 189, 977--987.

\bibitem[{Mano(2013)}]{man:2013}
Mano, S., 2013. Duality between the two-locus {Wright-Fisher} diffusion model
  and the ancestral process with recombination. Journal of Applied Probability
  50, 256--271.

\bibitem[{Neuhauser and Krone(1997)}]{neu:kro:1997}
Neuhauser, C., Krone, S.~M., 1997. The genealogy of samples in models with
  selection. Genetics 145, 519--534.

\bibitem[{{OEIS Foundation Inc.}(2011)}]{oeis:2011}
{OEIS Foundation Inc.}, 2011. The on-line encyclopedia of integer sequences.
\newline\urlprefix\url{http://oeis.org}

\bibitem[{Simonsen and Churchill(1997)}]{sim:chu:1997}
Simonsen, K.~L., Churchill, G.~A., 1997. A {Markov} chain model of coalescence
  with recombination. Theoretical Population Biology 52, 43--59.

\bibitem[{Stephens(2007)}]{ste:2007}
Stephens, M., 2007. Inference under the coalescent. In: Balding, D., Bishop,
  M., Cannings, C. (Eds.), Handbook of Statistical Genetics. Wiley, Chichester,
  UK, Ch.~26, pp. 878--908.

\bibitem[{Stephens and Donnelly(2003)}]{ste:don:2003}
Stephens, M., Donnelly, P., 2003. Ancestral inference in population genetics
  models with selection. Australia and New Zealand Journal of Statistics
  45~(3), 395--430.

\bibitem[{Wiuf and Hein(1997)}]{wiu:hei:1997}
Wiuf, C., Hein, J., 1997. On the number of ancestors to a {DNA} sequence.
  Genetics 147, 1459--1468.

\bibitem[{Wright(1949)}]{wri:1949}
Wright, S., 1949. Adaptation and selection. In: Jepson, G.~L., Mayr, E.,
  Simpson, G.~G. (Eds.), Genetics, Paleontology and Evolution. Princeton
  University Press, Princeton, pp. 365--389.

\end{thebibliography}


\end{document}